\documentclass{llncs}

\usepackage{amssymb,amsmath}
\usepackage[toc,page]{appendix}
\usepackage{placeins}
\usepackage{graphicx}
\usepackage{algorithm,algorithmic}
\usepackage{tikz}
\usetikzlibrary{plotmarks}
\usetikzlibrary{automata} 
\usetikzlibrary[automata] 
\usetikzlibrary{decorations}
\usepackage{xifthen}
\usepackage{pstricks,pst-node,pst-tree}
\usepackage{fancybox}
\usetikzlibrary{shapes.arrows,chains}
\usepackage[all]{xy}
\usepackage{bibentry}
\usepackage{graphics}
\usepackage{xspace,aeguill,ae}
\usepackage{verbatim,mathrsfs,enumerate,url}
\usepackage{nag}
\usepackage{wrapfig}

\makeatletter
\renewcommand*{\@fnsymbol}[1]{\ensuremath{\ifcase#1\or \star\or \dagger\or \ddagger\or
       \mathsection\or \mathparagraph\or \|\or **\or \dagger\dagger
       \or \ddagger\ddagger \else\@ctrerr\fi}}
\makeatother

\newcommand{\set}[1]{\left\lbrace #1 \right\rbrace} 

\newcommand{\IN}{\ensuremath{\mathbb{N}}\xspace}
\newcommand{\IR}{\ensuremath{\mathbb{R}}\xspace}
\newcommand{\IQ}{\ensuremath{\mathbb{Q}}\xspace}
\newcommand{\IZ}{\ensuremath{\mathbb{Z}}\xspace}

\newcommand{\lexicoleq}{\preccurlyeq}

\newcommand{\lexicol}{\prec}

\newcommand{\g}{{\sf Payoff}}
\newcommand{\RG}{\mathcal{G} = (V, V_1, V_2, E, r,\g)}

\newcommand{\lexicounRG}{(V, V_1, V_2, E, r, \g, \lexicoleq_1)}
\newcommand{\Glexun}{\G^{\lexicoleq_1}}
\newcommand{\Glexdeux}{\G^{\lexicoleq_2}}
\newcommand{\Glexi}{\G^{\lexicoleq_i}}

\newcommand{\Glexj}{\G^{\lexicoleq_j}}
\newcommand{\out}[1]{\left\langle #1 \right\rangle}

\DeclareMathOperator{\Min}{{\sf Min}}

\DeclareMathOperator{\MPInf}{{\sf InfMP}}
\DeclareMathOperator{\MPSup}{{\sf SupMP}}
\DeclareMathOperator{\InfVal}{\underline{Val}}
\DeclareMathOperator{\SupVal}{\overline{Val}}
\DeclareMathOperator{\Val}{Val}
\DeclareMathOperator{\M}{{\sf MP}}

\DeclareMathOperator{\InfVertex}{{\sf Inf}}

\DeclareMathOperator{\Compl}{{\sf C}}

\DeclareMathOperator{\LInf}{\liminf_{n \to \infty}}
\DeclareMathOperator{\LSup}{\limsup_{n \to \infty}}

\DeclareMathOperator{\LimSup}{{\sf LimSup}}
\DeclareMathOperator{\LimInf}{{\sf LimInf}}
\DeclareMathOperator{\Sup}{{\sf Sup}}
\DeclareMathOperator{\Inf}{{\sf Inf}}
\DeclareMathOperator{\Disc}{{\sf Disc}}

\newcommand{\G}{\mathcal{G}}
\newcommand{\A}{\mathcal{A}}

\newcommand{\D}{\mathcal{D}}

\newcommand{\CB}{\mathcal{C}}

\newcommand{\Safety}{\mathcal{S}}

\newcommand{\limAvgInf}{\liminf\limits_{n \to \infty} \frac{1}{n} \sum\limits_{k=0}^{n-1}}
\newcommand{\limAvgSup}{\limsup\limits_{n \to \infty} \frac{1}{n} \sum\limits_{k=0}^{n-1}}

\renewcommand{\phi}{\mathchar"11E}            

\newcommand{\bigO}[1]{\ensuremath{\mathop{}\mathopen{}O\mathopen{}\left(#1\right)}}

\title{Secure Equilibria in Weighted Games\thanks{Work partially supported by European project Cassting (FP7-ICT-601148).}}
\author{V\'{e}ronique Bruy\`{e}re\inst{1} \and No\'{e}mie Meunier\inst{1}$^,$\thanks{Author supported by F.R.S.-FNRS fellowship.} \and Jean-Fran\c cois Raskin\inst{2}$^{,}$\thanks{Author supported by ERC Starting Grant (279499: inVEST).}}
\institute{D\'{e}partement d'informatique, Universit\'{e} de Mons (UMONS), Belgium
\and 
D\'{e}partement d'informatique, Universit\'{e} Libre de Bruxelles (U.L.B.), Belgium}

\begin{document}
\maketitle
\begin{abstract}
We consider two-player non zero-sum infinite duration games played on weighted graphs. We extend the notion of secure equilibrium introduced by Chatterjee et al., from the Boolean setting to this quantitative setting. As for the Boolean setting, our notion of secure equilibrium refines the classical notion of Nash equilibrium. We prove that secure equilibria always exist in a large class of weighted games which includes common measures like sup, inf, lim sup, lim inf, mean-payoff, and discounted sum. Moreover we show that one can synthesize finite-memory strategy profiles with few memory. We also prove that the constrained existence problem for secure equilibria is decidable for sup, inf, lim sup, lim inf and mean-payoff measures.
Our solutions rely on new results for zero-sum quantitative games with lexicographic objectives that are interesting on their own right. 
\end{abstract}

\section{Introduction}

Two-player zero-sum infinite duration games played on graphs is a useful framework to formalize many important 
problems in computer science. 
System synthesis, and especially synthesis of reactive systems, is one of those important problems, see for example~\cite{RW87,PnueliR89}.
In this application, the vertices and the edges of the graph represent 
the states and the transitions of the system;
one player represents the system to synthesize, whereas the other represents the environment with which the system interacts.
In the classical setting, the environment is considered as {\em antagonist} and  
the objectives of the two players are {\em complementary}, leading to a so-called {\em zero-sum} game.
There are many known results about those zero-sum games~\cite{Mar75,EM79,EJ91,LNCS2500}. 

Modeling the environment as completely adversarial is usually a bold abstraction of reality. However it is often used as it is simple and sound: a winning strategy against a completely adversarial model of the environment, is winning against any environment which pursues its own objective. 
But this approach may fail to find a winning strategy whereas there exists a solution when the objective of the environment is taken into account.
In this case, we need to consider the more general framework of {\em non zero-sum} games.
The classical notion of rational behavior in this context is commonly formalized as Nash equilibria~\cite{Nas50}. 
Nash equilibria capture rational behaviors when the players only care about their own payoff (internal criteria), and they are indifferent to the payoff of the other player (external criteria). 
In the setting of synthesis, the more appropriate notion is the {\em adversarial
external criteria}, where the players are as harmful as possible to the other
players without sabotaging their own objectives. 
This has inspired the study of refinements of Nash equilibria, like the notion of 
{\em secure equilibria} that captures the adversarial external criteria and is at the basis of {\em compositional synthesis algorithms}~\cite{CHJ06}.
In secure equilibria, lexicographic objectives are considered: each player first tries to maximize his own payoff, and then tries to minimize the opponent's payoff. It is shown in~\cite{CHJ06} that secure equilibria are those Nash equilibria which represent enforceable contracts between the two players.

In this paper, we extend the notion of secure equilibria from the Boolean setting with $\omega$-regular objectives of~\cite{CHJ06} to a quantitative setting where objectives are non necessarily $\omega$-regular. More precisely, we consider two-player non zero-sum turn-based games played on weighted graphs, called \emph{weighted games} with objectives defined by classical measures considered in the literature for infinite plays: {\em sup}, {\em inf}, {\em lim sup}, {\em lim inf}, {\em mean-payoff}, and {\em discounted sum}~\cite{LaurentDoyen}. In our setting, the edges of the weighted graph are labelled with pairs of rational values that are used to assign two values to each infinite play: one value models the reward of player~1, and the other the reward of player~2. 

\paragraph{Contributions.}  Our contributions are threefold. (1) We show that all weighted games with the classical measures have {\em secure equilibria}. We also establish that there exist simple profiles of strategies that witness such equilibria: finite-memory strategies with a linear memory size are sufficient for most of the measures (polynomial size for inf and sup measures). (2) We provide polynomial or pseudo-polynomial (depending on the measures) algorithms for the automatic synthesis of such strategy profiles. We thus provide the necessary algorithms to extend the compositional synthesis framework of~\cite{ChatterjeeH07} to a quantitative setting. (3) We prove that one can decide the existence of a secure equilibrium whose outcome satisfies some constraints, for all measures except the discounted sum. In the latter case, we show that this problem is connected to a challenging open problem~\cite{Boker,ChatterjeeFW13}. Our solutions rely on the analysis of two-player zero-sum games with lexicographic objectives (for all the measures) for which we provide worst-case optimal algorithms. These so-called \emph{lexicographic games} are interesting on their own right and, to the best of our knowledge, were not studied in the literature.

\paragraph{Related work.}
Secure equilibria were first introduced in~\cite{CHJ06} in a Boolean setting, and for $\omega$-regular objectives. 
We extend this work to a quantitative setting in which objectives are not necessarily $\omega$-regular. 
Our contributions (1) and (2) are very much inspired by a recent work by Brihaye et al.~\cite{TJS} that provides general results for the existence of Nash equilibria in a large class of multi-player weighted graphs. We show that their main theorem is extendable to secure equilibria in the case of two-player games. However, to adapt their theorem we need non trivial new results on lexicographic games, while for Nash equilibria they can rely on well-known results for (non lexicographic) zero-sum two-player games.
Previously to~\cite{TJS}, existence results about different kinds of equilibria (Nash, secure, perfect, and subgame perfect) have been established for \emph{quantitative  reachability} objectives in multi-player weighted games~\cite{BrihayeBP10,BBDG,Kim}.
In \cite{UmmelsW11}, among other results, the authors study 
the decision problem of both the existence and the constraint existence of Nash equilibria in multi-player weighted games for the mean-payoff measure. A part of our contribution (3) is inspired by some of their techniques.
The {\em rational synthesis} problem is studied in~\cite{FKL10}, when a system interacts with agents that all have their own objectives. This problem asks to construct a strategy profile that enforces the objective of the system, and which is an equilibrium for the agents. The objectives are $\omega$-regular or defined by deterministic latticed B\"uchi word automata (they do not include mean-payoff and discounted sum objectives), and secure equilibria are not considered explicitly.
Lexicographic games were first considered in~\cite{KTB}, for a mean-payoff measure that differs from the one studied in this paper. The proof technique that we develop is similar to their approach but requires non trivial adaptations. To the best of our knowledge, lexicographic games for all the other objectives (sup, inf, lim sup, lim inf, and discounted sum) have not been studied previously.

\paragraph{Structure of the paper.}   In Section~\ref{sec:prelim}, we first recall classical notions on games and equilibria, we then introduce the three problems studied in this paper and our solutions for the sup, inf, lim sup, lim inf, mean-payoff, and discounted sum measures (Problems~\ref{prob:ES}-\ref{prob:ESconstraint} and Theorems~\ref{thm:existenceES}-\ref{thm:constraintES}). In Section~\ref{sec:equi}, for each of the three problems, we provide a general framework of weighted games in which the problem can be solved (Propositions~\ref{SE}-\ref{constSE}). Those frameworks impose in particular the determinacy of two (one for each player)  lexicographic games associated with the initial game. Lexicographic games are introduced in Section~\ref{sec:equi} and proved to be determined in Section~\ref{sec:lexico}, where their complexity is also established. With these results, in Section~\ref{sec:equi}, we are able to prove Theorems~\ref{thm:existenceES}-\ref{thm:complES}, because games with the considered measures all fall in the frameworks proposed for solving Problems~\ref{prob:ES}-\ref{prob:EScomplexity}. The framework proposed for the last problem requires an additional hypothesis that we study in Section~\ref{sec:path}. Theorem~\ref{thm:constraintES} can then be derived for all the measures, except the discounted sum. We also show that Problem~\ref{prob:ESconstraint} for this measure is linked to a challenging open problem~\cite{Boker,ChatterjeeFW13}. 
In  Section~\ref{sec:conclusion}, we give a conclusion.

\section{Weighted Games and Studied Problems} \label{sec:prelim}

In this section, we recall the notions of weighted game, Nash equilibrium, secure equilibrium, and we state the problems that we want to solve. 

\subsection{Weighted Games}

We here consider two-player turn-based non zero-sum weighted games such that the weights are seen as rewards, and the two players want to maximize their payoff (this payoff is computed from the weights, for example as in Definition~\ref{def:payoff}). 

\begin{definition}\label{RG}
A two-player non zero-sum \emph{weighted game} is a tuple $\G = (V, V_1,$ $V_2, E, r, \g)$ where 
\begin{itemize}

  \item $(V, E)$ is a finite directed graph, the arena of the game, with vertices $V$ and edges $E \subseteq V \times V$, such that for each $v \in V$, there exists $e = (v,v') \in E$ for some $v' \in V$ (no deadlock), 

  \item $V_1, V_2$ form a partition of $V$ such that $V_i$ is the set of vertices 
    controlled by player $i \in \set{1, 2}$, 

  \item $r = (r_1,r_2)$ is the weight function such that $r_i : E \to \IQ$ associates a rational reward with each edge for player $i \in \set{1, 2}$,
  
  \item $\g = (\g_1, \g_2)$ is the payoff function such that $\g_i  : V^\omega \to \IR$ is the payoff function of player $i \in \set{1, 2}$ such that $\g_i$ is defined starting from $r_i$.

\end{itemize}
\end{definition} 

When an initial vertex $v_0 \in V$ is fixed, we call $(\mathcal{G}, v_0)$ an \emph{initialized weighted game}. 
A \emph{play} of $(\mathcal{G}, v_0)$ is an infinite 
sequence $\rho = \rho_0 \rho_1 \ldots \in V^\omega$ 
such that $\rho_0 = v_0$ and  $(\rho_i, \rho_{i + 1}) \in E$ for all $i \in \IN$. 
\emph{Histories} of $(\mathcal{G}, v_0)$ are finite sequences $h = h_0 \ldots h_n \in V^+$ defined in the same way. 
A \emph{prefix} (resp. \emph{suffix}) of a play $\rho$ is a finite sequence $\rho_0 \dots \rho_n$ 
(resp. infinite sequence $\rho_n \rho_{n+1} \ldots$) denoted by $\rho_{\leq n}$ (resp. $\rho_{\geq n}$). 
The \emph{length} of $\rho_{\leq n}$ is the number $n$ of its edges. 
For a play $\rho \in V^\omega$ and a player 
$i \in \set{1, 2}$, we denote by 
$r_i(\rho) = r_i(\rho_0, \rho_1) r_i(\rho_1, \rho_2) \dots$ 
the sequence of player $i$ weights along $\rho$.
The \emph{payoff} of $\rho$ for player $i$ is given by $\g_i(\rho)$, 
and $\g(\rho) = (\g_1(\rho), \g_2(\rho))$ is the payoff of $\rho$ 
in $\G$. 

A \emph{strategy} $\sigma$ for player $i \in \set{1, 2}$ is a function 
$\sigma : V^\ast V_i \to V$ assigning to each history $hv \in V^\ast V_i$ 
a vertex $v' = \sigma(hv)$ such that $(v, v') \in E$. We denote by
$\Sigma_i$ the set of strategies of player $i$.
A strategy $\sigma$ for player $i$ is \emph{positional} 
if $\sigma(h) = \sigma(h')$ for all histories $h, h'$ ending with the same vertex ($\sigma$ only depends on the last vertex of the history). In particular, a positional strategy is a function  $\sigma : V_i \to V$ (instead of $\sigma : V^\ast V_i \to V$). 
A strategy $\sigma$ is a \emph{finite-memory} strategy if it only needs finite memory of the history (recorded by a finite strategy automaton). 
Formally a \emph{finite strategy automaton} for player $i \in \set{1, 2}$ over a weighted game 
$\RG$ is a Mealy automaton $\mathcal{M} = (M, m_0, V, \delta, \nu)$ where: 
\begin{itemize}

  \item $M$ is a non-empty, finite set of memory states, 

  \item $m_0 \in M$ is the initial memory state, 

  \item $\delta : M \times V \to M$ is the memory update function, 

  \item $\nu : M \times V_i \to V$ is the transition choice function for player $i$, 
    such that $(v, \nu(m, v)) \in E$ for all $m \in M$ and $v \in V_i$.

\end{itemize}
This automaton defines the finite-memory strategy $\sigma_{\mathcal{M}}$ such that $\sigma_{\mathcal{M}}(hv) = \nu(\hat{\delta}(m_0, h), v)$ for all $hv \in V^\ast V_{i}$, where $\hat{\delta}$ extends $\delta$ to histories starting from $m_0$ as expected. The memory size of $\sigma_{\mathcal{M}}$ is defined as the size of $M$.

Given a strategy $\sigma \in \Sigma_i$ with $i \in \set{1,2}$, we say that a play $\rho$ 
of $\mathcal{G}$ is \emph{consistent} with $\sigma$ if 
$\rho_{k + 1} = \sigma(\rho_{\leq k})$ for all $k \in \IN$ such that 
$\rho_k \in V_i$. 
A \emph{strategy profile} of $\mathcal{G}$ is a pair 
$(\sigma_1, \sigma_2)$ of strategies, with $\sigma_i \in \Sigma_i$ for each $i \in \set{1, 2}$. 
Given an initial vertex $v_0$, such a strategy profile determines a unique 
play of $(\mathcal{G}, v_0)$ that is consistent with both strategies. 
This play is called the \emph{outcome} of $(\sigma_1, \sigma_2)$ and is denoted 
by $\out{\sigma_1, \sigma_2}_{v_0}$. For a history $hv \in V^\ast V$, and a strategy profile $(\sigma_1, \sigma_2)$, we denote by $\out{{\sigma_1}|_h, \sigma_2}_v$ the outcome of $({\sigma_1}|_h, \sigma_2)$  in the initialized game $(\mathcal{G}, v)$, where ${\sigma_1}|_h$ is the strategy defined by $\sigma_1|_h (h'v') = \sigma_1(hh'v')$ for all histories $h'v' \in V^\ast V_1$ that begin with $v$. The outcome $\out{\sigma_1, {\sigma_2}|_h}_v$ is defined similarly. We say that a player
\emph{deviates} from a strategy (resp. from a play) if he does not carefully follow this strategy (resp. play).
We say that a strategy profile $(\sigma_1, \sigma_2)$ is a positional (resp. finite-memory) strategy profile 
if $\sigma_1$ and $\sigma_2$ are positional (resp. finite-memory) strategies. The memory size of a finite-memory strategy profile is equal to the maximum memory size of its strategies.

In this paper, we focus on several well-known payoff functions (see for instance~\cite{LaurentDoyen}). 

\begin{definition} \label{def:payoff}
Given a weighted game $\RG$, we define the payoff function $\g$ as one of the measures in $\{\Inf$, $\Sup$, $\LimInf$, $\LimSup$, $\MPInf$, $\MPSup$, $\Disc^\lambda \mbox{ for $\lambda \in \, ]0, 1[$} \}$, where for all $i \in \{1, 2\}$ and $\rho \in V^\omega$:
\begin{itemize}
  \item $\Inf_i(\rho)  = \inf_{n \in \IN} r_i(\rho_n, \rho_{n+1})$, 
  \item $\Sup_i(\rho) =  \sup_{n \in \IN} r_i(\rho_n, \rho_{n+1})$,
  \item $\LimInf_i(\rho) = \liminf\limits_{n \to \infty} r_i(\rho_n, \rho_{n+1})$, 
  \item$\LimSup_i(\rho)  =  \limsup\limits_{n \to \infty} r_i(\rho_n, \rho_{n+1})$,
  \item  $\MPInf_i(\rho)  =  \limAvgInf r_i(\rho_k, \rho_{k+1})$,  
  \item $\MPSup_i(\rho)  =  \limAvgSup r_i(\rho_k, \rho_{k+1})$,
  \item $\Disc^\lambda_i(\rho) =  (1 - \lambda) \cdot \sum_{n=0}^{\infty} \lambda^n r_i(\rho_n, \rho_{n+1})$. 
\end{itemize}
We also call these games \emph{$\g$ weighted games.}
\end{definition}

\subsection{Equilibria}

We now recall the concept of Nash equilibrium and secure equilibrium. 
In this aim we need to fix two lexicographic orders on $\IR^2$: 
a lexicographic order $\lexicoleq_1$ w.r.t. the first component 
and a lexicographic order $\lexicoleq_2$ w.r.t. the second component 
such that for all $(x_1, x_2), (x'_1, x'_2) \in \IR^2$, 
\begin{eqnarray*}
(x_1, x_2) \lexicoleq_1 (x'_1, x'_2) &\textrm{ iff }& (x_1 < x'_1) \lor (x_1 = x'_1 \land x_2 \geq x'_2),  \\ 
(x_1, x_2) \lexicoleq_2 (x'_1, x'_2) &\textrm{ iff }& (x_2 < x'_2) \lor (x_2 = x'_2 \land x_1 \geq x'_1). 
\end{eqnarray*} 
Notice that $(\IR^2, \lexicoleq_1)$ and $(\IR^2, \lexicoleq_2)$ are totally ordered sets.

A strategy profile $(\sigma_1, \sigma_2)$ in an initialized weighted game $(\G, v_0)$ is a Nash equilibrium if player 1 (resp. player 2)
has no incentive to deviate unilaterally from $\sigma_1$ (resp. $\sigma_2$), since he cannot strictly increase his payoff when using $\sigma'_1$ (resp. $\sigma'_2$)
instead of $\sigma_1$ (resp. $\sigma_2$). The notion of secure equilibrium is stronger in the sense that player $i$ has no incentive to deviate from $\sigma_i$
with respect to the order $\lexicoleq_i$ (instead of the usual order $\leq$ on his payoffs). 

\begin{definition} \label{def:ES}
Let $(\G, v_0)$ be an initialized weighted game. A strategy profile 
$(\sigma_1, \sigma_2)$ with $\sigma_i \in \Sigma_i$,  $i \in \set{1,2}$, is
a \emph{Nash equilibrium} in $(\G, v_0)$ if, for each strategy $\sigma'_i \in \Sigma_i$, $i \in \set{1,2}$,
\begin{eqnarray*}
\g_1(\out{\sigma'_1, \sigma_2}_{v_0}) &\leq& \g_1(\out{\sigma_1, \sigma_2}_{v_0}), \\ 
\g_2(\out{\sigma_1, \sigma'_2}_{v_0}) &\leq& \g_2(\out{\sigma_1, \sigma_2}_{v_0}).
\end{eqnarray*}
It is a \emph{secure equilibrium} in $(\G, v_0)$ if, for each strategy $\sigma'_i \in \Sigma_i$, $i \in \set{1,2}$,
\begin{eqnarray*}
\g(\out{\sigma'_1, \sigma_2}_{v_0}) &\lexicoleq_1& \g(\out{\sigma_1, \sigma_2}_{v_0}), \\ 
\g(\out{\sigma_1, \sigma'_2}_{v_0}) &\lexicoleq_2& \g(\out{\sigma_1, \sigma_2}_{v_0}). 
\end{eqnarray*}
\end{definition}

Note that the notion of secure equilibrium is a refinement of that of Nash equilibrium.
In a Nash equilibrium, each player only cares about his own payoff (he maximizes it), whereas in a secure equilibrium, he cares about the payoff of both players (he maximizes his payoff and then minimizes the payoff of the other player).

With the notation of Definition~\ref{def:ES}, we say that $\sigma'_1$ is a \emph{profitable deviation} for player 1 w.r.t. $(\sigma_1, \sigma_2)$ if  $\g_1(\out{\sigma_1, \sigma_2}_{v_0}) < \g_1(\out{\sigma'_1, \sigma_2}_{v_0})$ (resp. $\g(\out{\sigma_1, \sigma_2}_{v_0}) \lexicol_1 \g(\out{\sigma'_1, \sigma_2}_{v_0})$). Profitable deviations for player 2 are defined similarly.
With these terms, we can say that $(\sigma_1, \sigma_2)$ is a Nash equilibrium (resp. secure equilibrium) if no player has a profitable deviation w.r.t. $(\sigma_1, \sigma_2)$ for the relation $<$ (resp. $\lexicol_i$).

\begin{example}\label{Ex:ENandES}
Consider the initialized $\MPInf$ weighted game $(\G, v_0)$ depicted in Figure  \ref{Figure:WG}. 
Circle (resp. square) vertices are 
player $1$ (resp. player $2$) vertices, and the weights $(0, 0)$ are not specified. In this simple game, both players have two positional strategies that are respectively denoted by $\sigma_i$ and $\sigma'_i$ for each player $i \in \set{1, 2}$. These strategies are depicted in Figure \ref{Figure:WG}.
\begin{figure}[ht!]
\begin{center}
\begin{tikzpicture}[initial text=,auto, node distance=1cm, shorten >=1pt] 

\node[state, scale=0.7]    (v0)                           {$v_0$};
\node[state, scale=0.7]               (v1)    [below left=of v0]     {$v_1$};
\node[state, rectangle, scale=0.7]               (v2)    [below right=of v0]    {$v_2$};
\node[state, scale=0.7]    (v3)    [below left=of v2]     {$v_3$};
\node[state, scale=0.7]    (v4)    [below right=of v2]    {$v_4$};

\path[->] (v0) edge               node[left, scale=0.8] {$\sigma_1$} (v1)
               edge               node[right, scale=0.8] {$\sigma'_1$} (v2)

          (v1) edge [loop below]  node[midway, scale=0.7] {$(4,4)$}  ()

          (v2) edge               node[left, scale=0.8] {$\sigma_2$}  (v3)
               edge               node[right, scale=0.8] {$\sigma'_2$} (v4)

          (v3) edge [loop below]  node[midway, scale=0.7] {$(4,3)$}  ()

          (v4) edge [loop below]  node[midway, scale=0.7] {$(3,2)$}  ();

\end{tikzpicture}
\end{center}
\caption{A simple weighted game\label{Figure:WG}}
\end{figure}
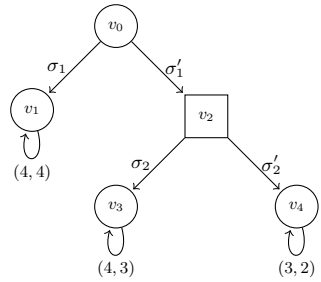

The strategy profiles $(\sigma_1, \sigma'_2)$ and $(\sigma'_1, \sigma_2)$ are secure equilibria in $(\G, v_0)$. Indeed, for the first one, only 
player~$1$ has control on the play $\out{\sigma_1, \sigma'_2}_{v_0}$ with payoff $(4,4)$, and he decreases his payoff if he plays strategy $\sigma'_1$ instead of $\sigma_1$ ($(3,2) \lexicol_1 (4,4)$). For the second one, both players have control on the play $\out{\sigma'_1, \sigma_2}_{v_0}$ with payoff $(4,3)$. If player~$1$ plays $\sigma_1$ instead of $\sigma'_1$, he keeps the same payoff but increases the payoff of player $2$ ($(4,4) \lexicol_1 (4,3)$). If player~2 plays $\sigma'_2$ instead of $\sigma_2$, he decreases his payoff ($(3,2) \lexicol_2 (4,3)$). Hence in both examples no player has a profitable deviation, and these strategy profiles are secure equilibria (and thus Nash equilibria) in $(\G, v_0)$. 

However the strategy profile $(\sigma_1, \sigma_2)$ is a Nash equilibrium which is not a secure equilibrium. 
Indeed, $\sigma'_1$ is a profitable deviation for player $1$ since by playing $\sigma'_1$ instead of $\sigma_1$, he keeps his own payoff but can decrease the payoff of player 2.

Finally, the strategy profile $(\sigma'_1, \sigma'_2)$ is neither a secure equilibrium nor a Nash equilibrium.
\end{example}

\subsection{Problems and Main Theorems}

To conclude this section, we state the three problems studied in the paper and our results. Let $(\G, v_0)$ be an initialized two-player non zero-sum weighted game. 

\begin{problem} \label{prob:ES}
Does there exist a secure equilibrium in $(\G, v_0)$? If it is the case, does there exist a finite-memory secure equilibrium?
\end{problem}

\begin{problem} \label{prob:EScomplexity}
What is the complexity of constructing a (finite-memory) secure equilibrium in $(\G, v_0)$ if one exists? 
\end{problem}

\begin{problem} \label{prob:ESconstraint}
Given two thresholds $\mu, \nu \in (\IQ \cup \{\pm \infty\})^2$, is it decidable whether there exists a secure equilibrium $(\sigma_1, \sigma_2)$ in $(\G, v_0)$ such that $$\mu \leq \g(\out{\sigma_1, \sigma_2}_{v_0}) \leq \nu,$$ i.e. $\mu_i \leq \g_i(\out{\sigma_1, \sigma_2}_{v_0}) \leq \nu_i$ for $i \in \{1,2\}$?
\end{problem}

As expected, if no restriction is given on the payoff function used in weighted games, the answer to Problem~\ref{prob:ES} is negative. An example of a weighted game with no Nash equilibrium (and thus with no secure equilibrium) is given in \cite{TJS}. 
In this paper, we solve the stated problems for weighted games with the payoff functions given in Definition~\ref{def:payoff}. Our solutions are as follows.

\begin{theorem} \label{thm:existenceES}
Let $(\G, v_0)$ be an initialized weighted game. Then $(\G, v_0)$ has a secure equilibrium with memory bounded by $|V| + 2$ (resp. $|V| \cdot |E|^2 + 3$) for payoff functions $\MPInf$, $\MPSup$, $\LimInf$, $\LimSup$ and  $\Disc^\lambda$ (resp. $\Inf$ and $\Sup$).
\end{theorem}

\begin{theorem} \label{thm:complES}
Let $(\G, v_0)$ be an initialized weighted game. Then one can compute a finite-memory secure equilibrium in $(\G, v_0)$ in pseudo-polynomial time (resp. polynomial time)  for payoff functions $\MPInf$, $\MPSup$ and  $\Disc^\lambda$ (resp. $\LimInf$, $\LimSup$, $\Inf$ and $\Sup$).
\end{theorem}

\begin{theorem} \label{thm:constraintES}
Let $(\G, v_0)$ be an initialized weighted game and $\mu, \nu \in (\mathbb{Q} \cup \{\pm \infty\})^2$ be two thresholds. Then one can decide in {\sf NP} $\cap$ {\sf co-NP} (resp. in ${\sf P}$) whether there exists a secure equilibrium $(\sigma_1, \sigma_2)$ in $(\G, v_0)$ such that $\mu \leq \g(\out{\sigma_1, \sigma_2}_{v_0}) \leq \nu$ for payoff functions $\MPInf$ and $\MPSup$ (resp. $\LimInf$, $\LimSup$, $\Inf$ and $\Sup$).
\end{theorem}

\noindent All these results are summarized in Tables~\ref{table:existence}-\ref{table:constrained}.

\begin{table}[h!]
\begin{center}
\begin{tabular}{|l|l|l|}
\hline
& Memory size & Construction time  \\
\hline 
$\MPInf$, $\MPSup$ &  $|V| + 2$ & pseudo-polynomial  \\
$\LimInf$, $\LimSup$ &  $|V| + 2$ & polynomial  \\
$\Inf$, $\Sup$ & $|V| \cdot |E|^2 + 3$ & polynomial \\
$\Disc^\lambda$ &  $|V| + 2$ & pseudo-polynomial  \\
\hline
\end{tabular}
\end{center}
\caption{Memory size and construction time of a finite-memory secure equilibrium (see Theorems~\ref{thm:existenceES}-\ref{thm:complES})}	\label{table:existence}
\end{table}
\begin{table}[h!]
\begin{center}
\begin{tabular}{|l|l|}
\hline
& Complexity \\
\hline 
$\MPInf$, $\MPSup$  & {\sf NP} $\cap$ {\sf co-NP} \\
$\LimInf$, $\LimSup$  & {\sf P} \\
$\Inf$, $\Sup$ & {\sf P} \\ 
$\Disc^\lambda$  & ? \\
\hline
\end{tabular}
\end{center}
\caption{Complexity of Problem~\ref{prob:ESconstraint} (see Theorem~\ref{thm:constraintES})}	\label{table:constrained}
\end{table}

\section{Lexicographic Payoff Games and Equilibria} \label{sec:equi}

To solve Problems~\ref{prob:ES}-\ref{prob:ESconstraint}, we follow the approach proposed in \cite{TJS} to solve the first problem for Nash equilibria (instead of secure equilibria). The general idea is the following one. Given an initialized (non zero-sum) weighted game $(\G, v_0)$, we derive two well-chosen two-player (zero-sum) games $\Glexun$ and $\Glexdeux$, and under adequate hypotheses, we obtain properties about secure equilibria in $(\G, v_0)$ through determinacy results and characterization of the optimal strategies of $\Glexun$ and $\Glexdeux$. 

In this section, we first introduce the games $\Glexun$ and $\Glexdeux$. Then for each of the three studied problems, we propose a general framework (i.e. adequate general hypotheses on $(\G, v_0)$, $\Glexun$ and $\Glexdeux$) under which we are able to solve the considered problem (Propositions~\ref{SE}-\ref{constSE}).

Later in Sections~\ref{sec:lexico} and~\ref{sec:path}, we will prove that these hypotheses are satisfied for most of the weighted games with the payoff functions of Definition~\ref{def:payoff}. We will thus be able to prove Theorems~\ref{thm:existenceES}-\ref{thm:constraintES} as a consequence of Propositions~\ref{SE}-\ref{constSE}.

\subsection{Lexicographic Payoff Games}

Let $\mathcal{G}$ be a weighted game as in Definition~\ref{RG}. We associate with $\G$ two zero-sum games $\G^1$ and $\G^2$, one for each player, that are respectively equipped with the lexicographic order $\lexicoleq_1$ and $\lexicoleq_2$ used to compare payoffs of these games.

\begin{definition} \label{def:lexgame}
From a weighted game $\RG$ as in Definition~\ref{RG}, we derive a zero-sum \emph{lexicographic payoff game} $\Glexun$ of the form $\lexicounRG$.
\end{definition} 

In the zero-sum game $\Glexun$, the two players have antagonistic goals. 
For each play $\rho \in V^\omega$,
player $1$ receives the \emph{payoff} $\g(\rho)$ that he wants to maximize w.r.t. the 
lexicographic order $\lexicoleq_1$, while player $2$ pays the \emph{cost} $\g(\rho)$ that 
he wants to minimize w.r.t. $\lexicoleq_1$. 
When $\g$ is one among the payoff functions of Definition~\ref{def:payoff}, we say that $\Glexun$ is a \emph{$\g$ lexicographic payoff game}. 

We also consider a dual lexicographic payoff game $\Glexdeux = (V, V_2, V_1, E, r,$ $\g, \lexicoleq_2)$ where the roles of the two players are exchanged and the used lexicographic order is $\lexicoleq_2$. In this game player $2$ wants to maximize $\g(\rho)$ w.r.t. $\lexicoleq_2$, while player $1$ wants to minimize it.

\begin{definition}
Given a lexicographic payoff game $\Glexun$, 
we define for every vertex $v \in V$ the \emph{upper value} $\SupVal(v)$ and the \emph{lower value} $\InfVal(v)$ respectively as:
\begin{eqnarray*}
\SupVal(v) &=&   \inf\limits_{\sigma_{2} \in \Sigma_{2}} \sup\limits_{\sigma_1 \in \Sigma_1} \g(\out{\sigma_1, \sigma_{2}}_v), \\
\InfVal(v) &=&  \sup\limits_{\sigma_1 \in \Sigma_1} \inf\limits_{\sigma_{2} \in \Sigma_{2}} \g(\out{\sigma_1, \sigma_{2}}_v). 
\end{eqnarray*}
The game $\Glexun$ is \emph{determined} if, 
for every $v \in V$, we have $\SupVal(v) = \InfVal(v)$. We also say 
that $\Glexun$ has a \emph{value} from $v$, and we write 
$\Val(v) = \SupVal(v) = \InfVal(v)$.
\end{definition}

In the previous definition, let us remind that the infimum and supremum functions are applied on a set of payoffs lexicographically ordered with
$\lexicoleq_1$. 

\begin{definition} \label{def:optimal}
Given a lexicographic payoff game $\Glexun$ and a vertex $v \in V$,
we say that $\sigma^\star_1 \in \Sigma_1$ is an \emph{optimal 
strategy} for player $1$ and vertex $v$ if, for each strategy $\sigma_{2} \in \Sigma_{2}$, 
we have 
\begin{equation*}
\InfVal(v) \lexicoleq_1  \g(\out{\sigma^\star_1, \sigma_{2}}_v).
\end{equation*} 
Similarly, $\sigma^\star_{2} \in \Sigma_{2}$ is an optimal strategy for player $2$ and vertex $v$ if, 
for each strategy $\sigma_1 \in \Sigma_1$, we have 
\begin{equation*}
\g(\out{\sigma_1, \sigma^\star_{2}}_v) \lexicoleq_1 \SupVal(v).
\end{equation*}
We say that $\Glexun$ is \emph{positionally-determined} if it is determined and has positional 
optimal strategies for both players and all vertices. Additionally, we say that $\Glexun$ is \emph{uniformly-determined} if the positional optimal strategies $\sigma^\star_1, \sigma^\star_2$ can be chosen globally independently of vertex $v$. We call these strategies \emph{uniform}.
\end{definition}

\begin{example}
We come back to the weighted game $\G$ of Figure~\ref{Figure:WG}, and consider the $\MPInf$ lexicographic payoff game $\Glexun$. Let us show that this game is uniformly-determined. Indeed each vertex has a value. In case of vertices $v_1, v_3$ and $v_4$, the value is trivially the weight of the loop on these vertices.  Vertex $v_2$ has value $(3,2)$ since the worse that player 2 can do is to play strategy $\sigma'_2$ ($(3,2) \lexicol_1 (4,3)$). Vertex $v_0$ has value $(4,4)$ (realised by strategy $\sigma_1$ of player 1). The positional optimal strategies are $\sigma_1$ and $\sigma'_2$ for player $1$ and player $2$ respectively, they are uniform. 
\end{example}

In the sequel, we will sometimes use the next lemma to prove that a game has a value from $v$ with optimal strategies $\sigma_1^\star, \sigma_2^\star$ for both players. 

\begin{lemma}\label{lemma:pointDeSelle}
If for all $v \in V$, there are a pair $(\alpha, \beta) \in \IR^2$, and two strategies $\sigma_1^\star \in \Sigma_1, \sigma_2^\star \in \Sigma_2$ such that for all 
$\sigma_1 \in \Sigma_1, \sigma_2 \in \Sigma_2$, 

\begin{equation}\label{pointDeSelle}
\g(\out{\sigma_1, \sigma_2^\star}_v) \lexicoleq_1 (\alpha, \beta) \lexicoleq_1 \g(\out{\sigma_1^\star, \sigma_2}_v), 
\end{equation}

then the game is determined and $\sigma_1^\star, \sigma_2^\star$ are optimal strategies for $v$ such that $\g(\out{\sigma_1^\star, \sigma_2^\star}_v) = (\alpha, \beta) = \Val(v)$. 

\end{lemma}

\begin{proof}
If we consider the strategies $\sigma_1 = \sigma_1^\star$ and $\sigma_2 = \sigma_2^\star$ in~(\ref{pointDeSelle}), we clearly have that $\g(\out{\sigma_1^\star, \sigma_2^\star}_v) = (\alpha, \beta)$. From Definition~\ref{def:optimal}, we have that $\InfVal(v) \lexicoleq_1 \g(\out{\sigma_1^\star, \sigma_2^\star}_v) \lexicoleq_1 \SupVal(v)$. Let us show that $\SupVal(v) \lexicoleq_1 (\alpha, \beta) \lexicoleq_1 \InfVal(v)$. It will follow that $\InfVal(v) = \SupVal(v) = \Val(v) = \g(\out{\sigma_1^\star, \sigma_2^\star}_v)$ and $\sigma_1^\star, \sigma_2^\star$ are optimal strategies.  

By (\ref{pointDeSelle}), we have 

\begin{align*}
(\alpha, \beta) = \inf_{\sigma_2 \in \Sigma_2} (\alpha, \beta) & \lexicoleq_1 \inf_{\sigma_2 \in \Sigma_2} \g(\out{\sigma_1^\star, \sigma_2}_v) \\	
 & \lexicoleq_1 \sup_{\sigma_1 \in \Sigma_1} \inf_{\sigma_2 \in \Sigma_2} \g(\out{\sigma_1, \sigma_2}_v) = \InfVal(v). 
\end{align*}

In a similar way, one can show that $\SupVal(v) \lexicoleq_1 (\alpha, \beta)$. 
\qed\end{proof}

Additionally to the notion of lexicographic payoff game, we need to define the next properties in a way to solve Problems~\ref{prob:ES}-\ref{prob:ESconstraint}.

\begin{definition}[\cite{TJS}] \label{def:pf}
Let $\RG$ be a weighted game. The payoff function $\g_i$, $i \in \set{1, 2}$, is \emph{prefix-linear} in $\G$ if, for every 
vertex $v \in V$ and history $hv \in V^+$, there exists 
$a \in \IR$ and $b \in \IR^+$ such that, for every play 
$\rho \in V^\omega$ whose first vertex is $v$, we have : 
\begin{equation*}
\g_i(h\rho) = a + b\cdot\g_i(\rho).
\end{equation*}
Moreover, $\g_i$ is \emph{prefix-independent} if for all $v \in V$, $hv \in V^+$ and  
$\rho \in V^\omega$ whose first vertex is $v$, we have $\g_i(h\rho) = \g_i(\rho)$.
\end{definition}

Clearly the notion of prefix-independent payoff function is a particular case of that of prefix-linear. 

\begin{remark}\label{rem:prefix-linear}
The two components of the payoff functions $\LimInf$, $\LimSup$, $\MPInf$ and $\MPSup$ are clearly prefix-independent. 

For payoff function $\Disc^{\lambda}$, the components are not prefix-independent but rather prefix-linear (see \cite{TJS}), since for $h = \rho_0 \ldots \rho_{n-1}$ and $\rho = \rho_n \rho_{n+1} \ldots$, we have for $i \in \{1,2\}$ that
$\Disc^{\lambda}_i(h \rho)  =  (1 - \lambda) \cdot \sum_{k=0}^\infty \lambda^k \cdot r_i(\rho_k, \rho_{k+1}) 
                      =  (1 - \lambda) \cdot \sum_{k=0}^{n-1} \lambda^k \cdot r_i(\rho_k, \rho_{k+1}) + \lambda^{n} \cdot \Disc^{\lambda}_i (\rho)$.
                      
Finally neither $\Inf$ nor $\Sup$ functions have prefix-linear components. Indeed, it is easy to find a history $h$ and two plays $\rho, \rho'$ such that $\Inf(h\rho) < \Inf(\rho), \Inf(\rho')$ with $\Inf(\rho) \ne \Inf(\rho')$, which implies that $a = \Inf(h\rho)$ and  $b=0$ in Definition~\ref{def:pf}. We then get a contradiction by taking $\rho''$ such that $\Inf(\rho'') < a$.
\end{remark}

\subsection{Existence of Secure Equilibria}

We can now state the general framework under which Problem~\ref{prob:ES} is positively solved. This framework is composed of all initialized weighted games $(\G,v_0)$ such that each lexicographic payoff game $\Glexi$ is uniformly-determined, and each payoff function $\g_i$ is prefix-linear, with $i \in \{1,2\}$.

\begin{proposition}\label{SE}
Let $\RG$ be a weighted game and $v_0$ be an initial vertex. If for each $i \in \set{1,2}$, $\g_i$ is prefix-linear, and $\Glexi$ is uniformly-determined, then
there exists a finite-memory secure equilibrium in $(\G,v_0)$ with memory size bounded by $|V| + 2$. 
\end{proposition} 

This proposition is an adaptation to secure equilibria in two-player games of a theorem\footnote{In Theorem $10$ of \cite{TJS}, one hypothesis is missing: determinacy must be uniform determinacy.} given in \cite{TJS,Julie} for the existence of Nash equilibria in multi-player games. The main difference is that we here need to work with (two-dimensional) lexicographic payoff games instead of classical (one-dimensional, non lexicographic) zero-sum quantitative games.  
The proof of this proposition is similar to the one given in \cite{TJS,Julie}. Nevertheless, we give this proof in a way to have a self-contained paper. In the next sections, we show that the hypotheses of Proposition~\ref{SE} are satisfied for most of payoff functions in Definition~\ref{def:payoff}.  

\begin{proof}
Let $\RG$ be a weighted game and let $v_0 \in V$ be an initial vertex. 
We know that the lexicographic payoff game $\Glexi$ is uniformly-determined for each $i \in \set{1, 2}$.
Let us fix some notation. In $\Glexi$, player $i$ wants to maximize his payoff against the other player who wants to minimize it.    
To emphasize this situation, we denote by player $-i$ (instead of player $3-i$) the player that is opposed to player $i$ in $\Glexi$. Moreover we denote by $\sigma_i^\star$ and $\sigma_{-i}^\star$ the uniform optimal strategies for player $i$ and $-i$ respectively in $\Glexi$. 
In other words, $\sigma^\star_1$ denotes a uniform optimal strategy of player~1 in $\Glexun$, and $\sigma^\star_{-1}$ denotes a uniform optimal strategy of player~2 in $\Glexun$. Similarly, $\sigma_2^\star$ (resp. $\sigma_{-2}^\star$) denotes a uniform optimal strategy of player $2$ (resp. player $1$) in $\Glexdeux$.  

We first show that there exists a secure equilibrium $(\tau_1, \tau_2)$ in $(\G,v_0)$, 
as follows: player $i$ plays according to his strategy 
$\sigma_i^\star$ in $\Glexi$, and punishes the other player $j$ 
if he is the first player to deviate from his strategy $\sigma_{j}^\star$ 
in $\Glexj$, by playing according to $\sigma_{-j}^\star$ in $\Glexj$. 

Formally, let $\rho = \out{\sigma_1^\star, \sigma_2^\star}_{v_0}$ be the outcome of the optimal strategies $(\sigma_1^\star, \sigma_2^\star)$ from $v_0$. We need to specify a punishment function 
$P : V^+ \to \set{\bot, 1, 2}$ that detects who is the first player to deviate from the play $\rho$, 
i.e. who has to be punished. For the initial vertex $v_0$, we define $P(v_0) = \bot$ and for every 
history $hv \in V^+ V$ starting in $v_0$, we let: 
$$
P(hv) = \left\{
  \begin{array}{lll}
    \bot & \mbox{if } P(h) = \bot \mbox{ and } hv \mbox{ is a prefix of } \rho, \\ 
    i & \mbox{if } P(h) = \bot, hv \mbox{ is not a prefix of } \rho, \mbox{ and } h \in V^\ast V_i, \\ 
    P(h) & \mbox{otherwise } (P(h) \neq \bot). 
  \end{array}
\right.
$$

For each player $i \in \set{1, 2}$ we define the strategy $\tau_i$ such that 
for all 
$hv \in V^\ast V_i$, 

\begin{equation} \label{eq:ES}
\tau_i(hv) = \left\{
  \begin{array}{ll}
    \sigma_i^\star(v) & \mbox{if } P(hv) = \bot \mbox{ or } i, \\ 
    \sigma_{-j}^\star(v) & \mbox{otherwise where } j = 3 - i \mbox{ is the other player}.
  \end{array}
\right.
\end{equation}

Clearly, the outcome of $(\tau_1, \tau_2)$ is the play $\rho = \out{\sigma^\star_1, \sigma^\star_2}_{v_0}$. 
Let us show that $(\tau_1, \tau_2)$ is a secure equilibrium in $(\G, v_0)$. 
We first prove that player $1$ has no profitable deviation. 
As a contradiction, let us assume that 
$\tau'_1$ is a profitable deviation for player $1$ w.r.t. $(\tau_1, \tau_2)$. 
We thus have that: 
\begin{equation}
\g(\rho) \lexicol_1 \g(\rho') \label{ProfitableDeviation} 
\end{equation}
with $\rho' = \out{\tau'_1, \tau_2}_{v_0}$.

Let $hv \in V^\ast V$ be the longest prefix common to $\rho$ and $\rho'$. This prefix exists and 
is finite as both plays $\rho$ and $\rho'$ start from vertex $v_0$ and $\rho \neq \rho'$. 
As the optimal strategies $\sigma_1^\star, \sigma_2^\star$ are uniform, we can write that $\rho = h \out{\sigma_1^\star, \sigma_2^\star}_v$. 
In the case of $\rho'$, player $1$ does not follow his strategy $\sigma_1^\star$ any more from vertex $v$, and so, player $2$ 
punishes him by playing according to his optimal strategy $\sigma_{-1}^\star$ after history $hv$. Therefore we have
$\rho' = h \out{{\tau'_1}|_h, \sigma_{-1}^\star}_v$. 

As the payoff functions are prefix-linear, there exist 
$a = (a_1,a_2) \in \IR^2$ and $b = (b_1,b_2) \in (\IR^+)^2$ such that 
\begin{eqnarray*}
\g(\rho') = a + b \cdot \g\left(\out{{\tau'_1}|_h, \sigma_{-1}^\star}_v\right), & and \\ 
& \\
\g(\rho) = a + b \cdot \g\left(\out{\sigma_1^\star, \sigma_2^\star}_v\right). &
\end{eqnarray*}
Since $\sigma_{-1}^\star$ is an optimal strategy for player $2$ in the lexicographic payoff game 
$\Glexun$, we have:
\begin{equation}
\g(\rho') ~{\lexicoleq}_1~ a + b \cdot \Val^1(v) \label{lexicoleq}, 
\end{equation}
where $\Val^1(v)$ is the value of $v$ in $\Glexun$. 
Furthermore, as $\sigma_1^\star$ is an optimal strategy for player $1$ in $\Glexun$, it follows that: 
\begin{equation}
a + b \cdot \Val^1(v) ~{\lexicoleq}_1~ \g(\rho) \label{lexicogeq}. 
\end{equation} 
From Equations (\ref{lexicoleq}) and (\ref{lexicogeq}) we have that 
$\g(\rho') \lexicoleq_1 \g(\rho)$ in contradiction with Equation (\ref{ProfitableDeviation}). 
This proves that player 1 has no profitable deviation.
We can show that player $2$ has not profitable deviation in the same way. 

We now prove that $(\tau_1, \tau_2)$ is a finite-memory strategy profile such that the memory size of both $\tau_1$ and $\tau_2$ is 
bounded by $|V| + 2$. For this purpose, we define a finite strategy 
automaton for both players that remembers the play $\rho$ and who has 
to be punished. As the play $\rho$ is the outcome of the uniform
strategy profile $(\sigma_1^\star, \sigma_2^\star)$, we can write 
$\rho = v_0 \dots v_{k-1} (v_k \dots v_n)^\omega$ where 
$0 \leq k \leq n \leq |V|$, $v_l \in V$ for all $0 \leq l \leq n$ 
and these vertices are all different. For any $i \in \set{1, 2}$, 
let $\mathcal{M}_i = (M, m_0, V, \delta, \nu)$ be the strategy automaton 
of player $i$, where\footnote{In this definition, vertex $v_{n+1}$ means vertex $v_k$.}: 
\begin{itemize}

  \item $M = \set{v_0v_0, v_0v_1, \dots, v_{n-1}v_n, v_nv_k} \cup \set{1,2}$. 
    With $M$ we remember the next edge that should be chosen to be sure that both players follow $\rho$, or the player $j \in \set{1,2}$ that has first deviated.
    
  \item $m_0 = v_0v_0$. This memory element means that the play has not begun yet.
  
  \item $\delta : M \times V \to M$ is defined for all $m \in M$ and $v \in V$ as follows: 
$$
\delta(m, v) = \left\{
    \begin{array}{lll}
        v_{l}v_{l+1} & \quad \mbox{if } m = uv_{l} \mbox{ and } v = v_l, \mbox{ with } u \in V , l \in \set{0,\dots,n} \\ 
          & \\
        j & \quad \mbox{if } m = j \in \set{1, 2} \mbox{or } (m = uv_{l}, v \neq v_l, \\ 
          & \quad \mbox{ with } u \in V_j, l \in \set{0, \dots, n})
            \end{array}
\right.
$$ 
Function $\delta$ udpates the memory either to the next edge of $\rho$ in case of no deviation, or to $j$ if player $j$ was the first to deviate.

  \item $\nu : M \times V_i \to V$ is defined for all $m \in M$ and $v \in V_i$ as follows:
$$
\nu(m, v) = \left\{
    \begin{array}{lll}
        \sigma_i^\star(v) & \quad \mbox{if } \delta(m,v) =  v_{l}v_{l+1}, \mbox{ with } l \in \set{0,\dots,n} \\ 
        \sigma_{-j}^\star(v) & \quad \mbox{if } \delta(m,v) =  j   \in \set{1, 2}         \end{array}
\right.
$$ 
Function $\nu$ proposes to play according to Equation~(\ref{eq:ES}).
\end{itemize}

Obviously, the strategy $\sigma_{\mathcal{M}_i}$ computed by the strategy automaton $\mathcal{M}_i$ 
exactly corresponds to the strategy $\tau_i$ of the secure equilibrium. Notice that in the definition of $M$, we can forget the memory element $i$. Indeed if player $i$ follows the strategy computed by $\mathcal{M}_i$, then he never deviates from the play $\rho$. Therefore, $M$ has size at most  $|V| + 2$, and strategy $\tau_i$ requires a memory of size at most $|V| + 2$. 
\qed\end{proof}

\begin{example}
We come back again to the weighted game $\G$ of Figure~\ref{Figure:WG} with the initial vertex $v_0$. Uniform optimal strategies in game $\Glexun$ are $\sigma^\star_1 = \sigma_1$ and  $\sigma^\star_{-1} = \sigma'_2$ (we use the notations of the previous proof). In $\Glexdeux$, they are $\sigma^\star_{-2} = \sigma'_1$ and  $\sigma^\star_{2} = \sigma_2$. We thus obtain the secure equilibrium $(\tau_{1},\tau_{2})$ leading to the outcome $\rho = \out{\sigma_1, \sigma_2}_{v_0}$. If player $1$ deviates from $\rho$ (in order to decrease player $2$ payoff), then player $2$ punishes him by playing $\sigma'_2$ which is the worst for player~$1$ (since it decreases his payoff). 
\end{example}

Notice that the proof of Proposition~\ref{SE} also holds for payoff functions that mix different measures for the two players, like for example $\MPInf_1$ for player~1 and $\Sup_2$ for player~2.

We could also prove a result similar to Proposition~\ref{SE} such that the hypotheses are replaced by the following ones: $(\G,v_0)$ is an initialized weighted game such that for each $i \in \set{1,2}$, $\Glexi$ is determined (less restrictive than uniformly-determined) and $\g_i$ is prefix-independent (more restrictive than prefix-linear). 
The same kind of result is proved in \cite{Julie} in the case of Nash equilibria; its proof can be easily adapted to the case of secure equilibria as done for the proof of Proposition~\ref{SE}. 

\subsection{Construction of Secure Equilibria}

In the previous section, we have proposed  a rather general framework of weighted games for which Problem~\ref{prob:ES} has a positive answer. In this section we show that we can solve Problem~\ref{prob:EScomplexity} in the same general framework, with the additional hypothesis that one can compute uniform optimal strategies in $\Glexi$, for $i \in \set{1,2}$.

\begin{proposition}\label{complSE}
Let $(\G,v_0)$ be an initialized weighted game such that for each $i \in \set{1,2}$, $\g_i$ is prefix-linear, and $\Glexi$ is uniformly-determined with computable uniform optimal strategies. If such strategies can be computed in $\Compl$~time for both players, then a finite-memory secure equilibrium in $(\G,v_0)$ can also be constructed  in $\Compl$~time.
\end{proposition} 

\begin{proof}
By Proposition~\ref{SE}, we know that there exists a finite-memory secure equilibrium $(\tau_1,\tau_2)$ in $(\G,v_0)$. Moreover the proof of this theorem indicates how to construct $(\tau_1,\tau_2)$ from uniform optimal strategies in $\Glexun$ and $\Glexdeux$ (see (\ref{eq:ES}) and the finite strategy automaton proposed for both players at the end of this proof). It follows that a finite-memory secure equilibrium can be constructed in $\Compl$~time.
\qed\end{proof}

\subsection{Constrained Existence of Secure Equilibria}

Let us turn to Problem~\ref{prob:ESconstraint}. 
First notice that contrarily to Problem~\ref{prob:EScomplexity}, strategy profiles for solving this problem require infinite memory, as shown by the next example.

\begin{example} \label{ex:infini}
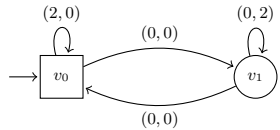
\begin{figure}[ht!]
\centering
\begin{tikzpicture}[initial text=,auto, node distance=2cm, shorten >=1pt] 

\node[state, initial, rectangle, scale=0.7]    (v0)                           {$v_0$};
\node[state, scale=0.7]    (v1)    [right=of v0]     {$v_1$};

\path[->] (v0) edge [bend left=25]              node[midway, scale=0.7] {$(0, 0)$} (v1)
               edge [loop above]                 node[midway, scale=0.7] {$(2, 0)$} ()

          (v1) edge [bend right=-25]              node[midway, scale=0.7] {$(0, 0)$}  (v0)
               edge [loop above]                 node[midway, scale=0.7] {$(0, 2)$} ();

\end{tikzpicture}
\centering
\caption{An weighted game that shows the need of infinite memory for Problem~\ref{prob:ESconstraint}} \label{fig:infini}
\end{figure}

Consider the initialized $\MPInf$ weighted game $(\G,v_0)$ of Figure~\ref{fig:infini}. Take thresholds $\mu = (1,1)$ and $\nu = (+\infty, +\infty)$. 

As explained in~\cite[Proof of Lemma 7]{Velner}, any strategy profile that is finite-memory produces an ultimately periodic outcome $\rho$ such that $\g_1(\rho) + \g_2(\rho) < 1$, thus not satisfying threshold $(1,1)$. However there exists an infinite path $\rho^{\ast}$ in $(\G,v_0)$ with payoff equal to $(1,1)$: $\rho^{\ast}$ visits $n$ times vertex $v_0$  and then $n$ times vertex $v_1$, and repeats this forever with increasing values of $n$.

Let us define a strategy profile $(\sigma_1,\sigma_2)$ such that its outcome is equal to $\rho^{\ast}$, and as soon as one player deviates from $\rho^{\ast}$, the other player punishes him by always chosing an edge with weights $(0,0)$. This strategy  profile is a secure equilibrium since one verifies that the player who first deviates receives a payoff of 0 instead of 1.

Therefore there exists a secure equilibrium $(\sigma_1,\sigma_2)$ in $(\G,v_0)$ that satisfies $\mu \leq \g(\out{\sigma_1, \sigma_2}_{v_0}) \leq \nu$, but it cannot be finite-memory.
\end{example}

We show that Problem~\ref{prob:ESconstraint} is decidable in the same framework as for Problem{~\ref{prob:ES}, with the additional hypotheses that one can compute values in $\Glexi$, for $i \in \set{1,2}$, and the next problem is decidable.

\begin{problem} \label{prob:path}
Let $G = (V,E,v_0,r,\Val)$ be a finite directed graph with an initial vertex $v_0$, a weight function $r = (r_1, r_2)$ with $r_i : E \to \IQ$, and a value function $\Val = (\Val^1,\Val^2)$, with $\Val^i : V \to \IQ^2$. Let $\mu, \nu \in (\IQ \cup \{\pm \infty\})^2$ be two thresholds. Is it decidable whether there exists an infinite path $\rho$ in $G$ starting in $v_0$ such that 
\begin{itemize}
\item $\forall k \geq 0, \forall i \in \{1,2\}, \Val^i(\rho_k) \lexicoleq_i \g(\rho_{\geq k})$, and
\item $\mu \leq \g(\rho) \leq \nu$?
\end{itemize}
\end{problem}

\begin{proposition} \label{constSE}
Let $(\G,v_0)$ be an initialized weighted game and $\mu, \nu \in (\IQ \cup \{\pm \infty\})^2$ be two thresholds. Suppose that 
\begin{itemize}
\item for each $i \in \set{1,2}$, $\g_i$ is prefix-linear, and $\Glexi$ is uniformly-determined with computable values,  
\item Problem~\ref{prob:path} is decidable for the graph $G$ constructed from $\G$ such that both functions $\Val^i$, $i \in \{1,2\}$, are constructed from the values $\Val^i(v)$, $v \in V$, in $\Glexi$.
\end{itemize}
One can decide whether  there exists a secure equilibrium $(\sigma_1, \sigma_2)$ in $(\G, v_0)$ such that $\mu \leq \g(\out{\sigma_1, \sigma_2}_{v_0}) \leq \nu$.
\end{proposition}

The proof of this theorem is based on the next lemma that characterizes the outcomes of secure equilibria in $(\G,v_0)$.

\begin{lemma} \label{lem}
Let $\rho$ be a play in $\G$ starting in $v_0$.  Then $\rho$ is the outcome of a secure equilibrium in $(\G,v_0)$ if and only if  
\begin{eqnarray} \label{eq:caractES}
\forall k \geq 0, \forall i \in \{1,2\}, \quad \Val^i(\rho_k) \lexicoleq_i \g(\rho_{\geq k}).
\end{eqnarray}	
\end{lemma}

\begin{proof}
As in the proof of Proposition~\ref{SE}, we denote by $\sigma_i^\star$ and $\sigma_{-i}^\star$ the uniform optimal strategies for player $i$ and $-i$ respectively in $\Glexi$.

First, let $(\tau_1,\tau_2)$ be a secure equilibrium in $(\G,v_0)$, and $\rho = \out{\tau_1, \tau_2}_{v_0}$ be its outcome. By contradiction, assume that there exists $k \geq 0$ and $i \in \{1,2\}$ such that $\g(\rho_{\geq k})  \lexicol_i \Val^i(\rho_k)$. Take such a $k$ as small as possible and let $h = \rho_{\leq k-1}$. Suppose that $i = 1$ (the case $i=2$ is similar). We can then construct a profitable deviation $\tau'_1$ for player~1: he follows the strategy $\tau_1$ until vertex $\rho_k$ from which he uses his optimal strategy $\sigma_1^\star$ in $\Glexun$. As $\g$ is prefix-linear and $\sigma_1^\star$ is uniform, it follows that for some $a \in \IR^2$ and $b \in (\IR^+)^2$, we have
\begin{eqnarray*}
\g(\out{\tau_1, \tau_2}_{v_0}) &=& a + b \cdot \g(\rho_{\geq k}) \\
							   &\lexicol_1& a + b \cdot \Val^1(\rho_k) \\
							   &\lexicoleq_1& a + b \cdot \g(\out{\sigma_1^\star, \tau_2|_{h}}_{\rho_k}) \\
							   &=& \g(\out{\tau'_1, \tau_2}_{v_0}).
\end{eqnarray*}							   
This is impossible since  $(\tau_1,\tau_2)$ is a secure equilibrium.

Next, let $\rho$ be a play that starts with $v_0$ and satisfies ($\ref{eq:caractES}$). We define a strategy profile $(\tau_1,\tau_2)$ such that $\out{\tau_1, \tau_2}_{v_0} = \rho$, and as soon as player~$i$ deviates, the other player uses his strategy $\sigma_{-i}^\star$ in $\G^{\lexicoleq_i}$ to punish him. Let us prove that $(\tau_1,\tau_2)$ is a secure equilibrium in $(\G,v_0)$. Let $\tau'_1$ be a strategy for player~$1$ such that $\rho' = \out{\tau'_1, \tau_2}_{v_0} \neq \rho$, and let $h\rho_k$ be the longest common prefix of $\rho$ and $\rho'$. We have for some $a \in \IR^2$ and $b \in (\IR^+)^2$:
\begin{eqnarray*}
 \g(\out{\tau'_1, \tau_2}_{v_0}) &=& a + b \cdot \g(\out{\tau'_1|_{h},\sigma_{-1}^\star}_{\rho_k}) \\
 							    &\lexicoleq_1& a + b \cdot \Val^1(\rho_k) \\
							    &\lexicoleq_1& a + b \cdot \g(\rho_{\geq k}) \\
							   &=& \g(\rho).
\end{eqnarray*}		
This shows that player~1 has no profitable deviation. The same holds for player~2. Hence $(\tau_1,\tau_2)$ is a secure equilibrium.
\qed\end{proof}

\begin{proof}[of Proposition~\ref{constSE}]
We propose the next algorithm. First compute $\Val^i(v)$ for each vertex $v$ of $\Glexi$, $i \in \{1,2\}$. Then construct from $(\G,v_0)$ the graph $G = (V,E,v_0,r,\Val)$. Finally test whether there exists a path $\rho$ in $G$ starting in $v_0$ such that $\forall k \geq 0, \forall i \in \{1,2\}, \Val^i(\rho_k) \lexicoleq_i \g(\rho_{\geq k})$, and $\mu \leq \g(\rho) \leq \nu$. Notice that it is indeed an algorithm since each $\Glexi$ is computationally uniformly-determined, and Problem~\ref{prob:path} is decidable. This algorithm is correct by Lemma~\ref{lem}.
\qed\end{proof}

Again, as for Propositions~\ref{SE}-\ref{complSE}, Proposition~\ref{constSE} also holds for distinct payoff functions that mix different measures for the two players

\section{Determinacy of Lexicographic Payoff Games} \label{sec:lexico}

In Section~\ref{sec:equi}, we have established strong links between secure equilibria in initialized weighted games $(\G,v_0)$ and determinacy of the two lexicographic payoff games $\Glexi$, $i \in \set{1,2}$. We have shown how to solve Problems~\ref{prob:ES}-\ref{prob:ESconstraint} under some adequate general hypotheses (see Propositions~\ref{SE}-\ref{constSE}).
In this section, we study the determinacy of lexicographic payoff games for the payoff functions proposed in Definition~\ref{def:payoff}. We show that for all payoffs, these games are uniformly-determined, except for the $\Inf$ and $\Sup$ payoffs for which they are only positionally-determined. We also study the complexity of computing values and optimal strategies for these games. Besides being important ingredients to solve Problems~\ref{prob:ES}-\ref{prob:ESconstraint}, these results are also very interesting on their own right. Our results are the following ones (we state them only for $\Glexun$):

\begin{theorem} \label{thm:uniform-det}
Let $\Glexun$ be a lexicographic payoff game. Then $\Glexun$ is uniformly-determined (resp. positionally-determined) for payoff functions $\MPInf$, $\MPSup$, $\LimInf$, $\LimSup$ and  $\Disc^\lambda$ (resp. $\Inf$ and $\Sup$).
\end{theorem}

\begin{theorem} \label{thm:compl-lexico}
Let $\Glexun$ be a lexicographic payoff game. 
\begin{enumerate}
\item Let $v$ be a vertex of $\Glexun$ and $(\alpha,\beta) \in \IQ^2$ be a pair of rationals. Deciding whether $(\alpha,\beta) \lexicoleq_1 \Val(v)$  is in {\sf NP} $\cap$ {\sf co-NP} for payoff functions $\MPInf$, $\MPSup$ and  $\Disc^\lambda$,  and {\sf P}-complete for  payoff functions $\LimInf$, $\LimSup$, $\Inf$ and $\Sup$.
\item The (rational) value of each vertex of $\Glexun$ can be computed in pseudo-polynomial time for payoff functions $\MPInf$, $\MPSup$ and  $\Disc^\lambda$, and in polynomial time for  payoff functions $\LimInf$, $\LimSup$, $\Inf$ and $\Sup$.
\item Uniform (resp. positional) optimal strategies for both players can be computed in pseudo-polynomial time for payoff functions $\MPInf$, $\MPSup$ and  $\Disc^\lambda$, and in polynomial time for payoff functions $\LimInf$ and $\LimSup$ (resp. $\Inf$ and $\Sup$).\end{enumerate}
\end{theorem}

\noindent All the results of Theorems~\ref{thm:uniform-det}-~\ref{thm:compl-lexico} are summarized in Table~\ref{table:determined}.

\begin{table}[ht!]
\begin{center}
\begin{tabular}{|l|l|l|l|l|}
\hline
 & determinacy & $(\alpha,\beta) \lexicoleq_1 \Val(v)$?  & values & optimal strategies \\
\hline 
$\MPInf$, $\MPSup$ & uniform & {\sf NP} $\cap$ {\sf co-NP} & pseudo-polynomial & pseudo-polynomial \\
$\LimInf$, $\LimSup$ & uniform & {\sf P}-complete & polynomial & polynomial \\
$\Inf$, $\Sup$ &  positional & {\sf P}-complete & polynomial & polynomial \\
$\Disc^\lambda$ & uniform & {\sf NP} $\cap$ {\sf co-NP} & pseudo-polynomial & pseudo-polynomial \\
\hline
\end{tabular}
\end{center}
\caption{Obtained results for lexicographic games (see Theorems~\ref{thm:uniform-det}-\ref{thm:compl-lexico})}	\label{table:determined}
\end{table}

Theorem~\ref{thm:uniform-det} states that every $\Inf$ lexicographic payoff game is positionally-determined. The next example shows that there is no hope to have uniformly-determined such games.

\begin{example} \label{ex}
Consider the $\Inf$ lexicographic payoff game depicted in Figure~\ref{Figure:InfLP}. 
Clearly, this game is positionally-determined, such that each vertex has value $(2, 0)$ except vertex $v_4$ that has value $(3, 1)$. Let us show that it is not uniformly-determined. To guarantee value $(2, 0)$ for $v_0$, player 1 has to use the positional strategy $\sigma_1^\star$ such that $\sigma_1^\star(v_4) = v_3$.   Indeed, this is the only way to have  
$(2, 0) \lexicoleq_1 \Inf(\out{\sigma_1^\star, \sigma_2}_{v_0})$ for all strategy $\sigma_2$ of player $2$. However, with this strategy, player~1 cannot guarantee value $(3,1)$ for $v_4$.

\begin{figure}[ht!]
\begin{center}
\begin{tikzpicture}[initial text=,auto, node distance=1.5cm, shorten >=1pt] 

\node[initial, state, scale=0.7]    (v0)                           {$v_0$};
\node[state, rectangle, scale=0.7]               (v2)    [right=of v0]    {$v_2$};
\node (fictif) [right=2cm of v2] {};
\node[state, scale=0.7]    (v3)    [above=0.5cm of fictif]     {$v_3$};
\node[state, scale=0.7]    (v4)    [below=0.5cm of fictif]    {$v_4$};

\path[->] (v0) edge               node[midway, scale=0.7] {$(2, 1)$} (v2)

          (v2) edge               node[above, scale=0.7] {$(2, 0)$}  (v3)
               edge               node[below, scale=0.7] {$(3, 1)$} (v4)

          (v4) edge               node[right, scale=0.7] {$(2, 0)$}  (v3)

          (v3) edge [loop right]  node[midway, scale=0.7] {$(2, 0)$}  ()

          (v4) edge [loop right]  node[midway, scale=0.7] {$(3, 1)$}  ();

\end{tikzpicture}
\end{center}
\caption{An $\Inf$ lexicographic payoff game\label{Figure:InfLP} that is not uniformly-determined.}
\end{figure}
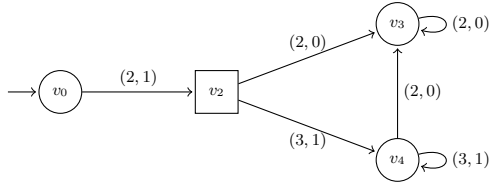

\end{example}

Before giving the proofs of Theorems~\ref{thm:uniform-det} and~\ref{thm:compl-lexico}, notice that we are now able to prove Theorems~\ref{thm:existenceES}-\ref{thm:complES} for the payoff functions $\MPInf$, $\MPSup$, $\LimInf$, $\LimSup$ and $\Disc^\lambda$. Indeed the hypotheses of  Propositions~\ref{SE}-\ref{complSE} are satisfied as the two components of these payoff functions are prefix-linear (see Remark~\ref{rem:prefix-linear}) and the games $\Glexi$, $i \in \{1,2\}$, are uniformly-determined with computable uniform optimal strategies (see Theorems~\ref{thm:uniform-det}-\ref{thm:compl-lexico}). 


\begin{proof}[of Theorems~\ref{thm:existenceES} and~\ref{thm:complES}, except for $\Inf$ and $\Sup$ payoffs] 
Let $(\G,v_0)$ be an initialized weighted game with a payoff function equal to one among $\MPInf$, $\MPSup$, $\LimInf$, $\LimSup$ and $\Disc^\lambda$.
By Remark~\ref{rem:prefix-linear}, this function is prefix-linear. Moreover, by Theorem~\ref{thm:uniform-det}, each game $\Glexi$, $i \in \set{1,2}$, is uniformly-determined. Therefore, since the hypotheses of Proposition~\ref{SE} are satisfied, there exists a finite-memory secure equilibrium in $(\G,v_0)$ with memory size at most $|V| + 2$, leading to Theorem~\ref{thm:existenceES}. The complexities stated in Theorem~\ref{thm:complES} are obtained as a consequence of Propositions~\ref{complSE} and~\ref{thm:compl-lexico}.
\qed\end{proof}

For $\Inf$ (resp. $\Sup$) payoffs, the proof of Theorem~\ref{thm:existenceES} cannot be based on Proposition~\ref{SE} since none of its hypotheses is satisfied. Indeed, the payoff function $\Inf$ is not prefix-linear (see Remark~\ref{rem:prefix-linear}), and there exist $\Inf$ lexicographic payoff games that are not uniformly-determined (see Example~\ref{ex}). 
However we get Theorems~\ref{thm:existenceES}-\ref{thm:complES} for $\Inf$ (resp. $\Sup$) payoffs thanks to a polynomial reduction to $\LimInf$ (resp. $\LimSup$) payoffs and the fact that Theorems~\ref{thm:existenceES}-\ref{thm:complES} hold for the latter payoffs.


\begin{proof}[of Theorems~\ref{thm:existenceES} and~\ref{thm:complES} for $\Inf$ and $\Sup$ payoffs] 
Let $\G = (V, V_1, V_2, E, r, \Inf)$ be an $\Inf$ weighted game and $v_0$ be an initial vertex (the proof is similar for $\Sup$ payoff). Let $R = \{r_1(e) \mid e \in E\} \cup \{r_2(e) \mid e \in E\} \cup \{+\infty\}$ be the set of weights labeling the edges of $E$ with, in addition, a highest new weight $+\infty$. We construct 
an initialized $\LimInf$ weighted game $(\G', v'_0)$ with $\G' = (V', V'_1, V'_2, E', r', \LimInf)$ as follows:
\begin{itemize}

  \item $V' = V \times R \times R$ is partitioned into $V'_1 = V_1 \times R \times R$ and $V'_2 = V_2 \times R \times R$,

  \item $E' \subseteq V' \times V'$ is the set of edges of $E$ augmented with the current smallest weights, that is, $e' = ((v, w_1, w_2), (v', w'_1, w'_2)) \in E'$ iff $e = (v,v') \in E$ and $w'_i = \min \{w_i,r_i(e)\}$, with $i \in \set{1, 2}$,
  
    \item $r' = (r'_1,r'_2)$ is the weight function that also remembers the current smallest weights, that is, for each $i \in \set{1, 2}$,  $r'_i(e') = w'_i$ for each edge $e' = ((v, w_1, w_2), (v', w'_1, w'_2)) \in E'$,

  \item $v'_0 = (v_0, +\infty, +\infty)$ is the initial vertex augmented with the highest weight $+\infty$ for both components. 

\end{itemize}

Clearly, to any play $\rho$ in $(\G,v_0)$ corresponds a unique play $\rho'$ in $(\G',v'_0)$, and conversely.
Moreover, by construction of $\G'$, for all $n > 0$, we have $\rho'_n = (\rho_n,w^1_n,w^2_n)$ with 
$$w^i_n = \min\limits_{0 \leq k < n} r_i((\rho_k, \rho_{k+1}))$$
for $i \in \set{1,2}$. 
Since there is a finite number of possible weight values, the decreasing weight components of $\rho'$ eventually stabilise along the play. It follows that $\Inf(\rho) = \LimInf(\rho')$.

Clearly any finite-memory secure equilibrium in $(\G,v_0)$ is a finite-memory secure equilibrium in $(\G',v'_0)$, and conversely. By Theorems~\ref{thm:existenceES} and~\ref{thm:complES} for $\LimInf$ payoff, $(\G', v'_0)$ has a secure equilibrium with memory at most $|V'| + 2$  that can be computed in polynomial time. Since $|V'| \leq |V| \cdot |E|^2 + 1$, we can conclude that $(\G, v_0)$ has a secure equilibrium with memory at most $|V| \cdot |E|^2 + 3$ that can also be computed in polynomial time. 
\qed\end{proof}

In the next sections~\ref{subsec:MP}-\ref{subsec:Disc}, we are going to prove Theorems~\ref{thm:uniform-det} and~\ref{thm:compl-lexico} for the various payoff functions. In this aim, let us first briefly recall some useful notions and results about several classes of one-dimensional (instead of two-dimensional) games. 

\subsection{Useful Background}

\emph{Mean-payoff games} (resp. \emph{Liminf games}, \emph{discounted games}) are games as in Definition~\ref{def:lexgame} except that functions $r$ and $\g = \MPInf$ (resp. $\LimInf$, $\Disc^\lambda$) are one-dimensional instead of two-dimensional. The comparison of two payoffs is thus done with usual order $\leq$ instead of $\lexicoleq_1$. As for lexicographic payoff games, the notions of value and optimal strategies can be defined for such games.
It is a classical result that mean-payoff games are uniformly-determined~\cite{EM79,ZP}. Moreover deciding whether the value of a vertex is greater than or equal to a given threshold is in in {\sf NP} $\cap$ {\sf co-NP}, there exists a pseudo-polynomial time algorithm to compute the values of each vertex and optimal strategies for both players~\cite{ZP}. The same results hold for discounted games~\cite{ZP,LNCS2500}. LimInf games  are known to be uniformly-determined~\cite{DH09}.

In \cite{GimbertZ05}, the authors have proposed general conditions to guarantee the uniform determinacy of zero-sum games $\G^\sqsubseteq = (V, V_1, V_2, E, r, \g, \sqsubseteq)$ where $\g$ is any payoff function and $\sqsubseteq$ is any preference relation\footnote{A preference relation is a complete, reflexive and transitive relation.}

\begin{theorem}[\cite{GimbertZ05}]\label{thm:UniformlyDeterminedHugo}
Suppose that for the preference relation $\sqsubseteq$ each zero-sum game $\G^\sqsubseteq = (V, V_1, V_2, E, r, \g, \sqsubseteq)$ such that either $V_1 = \emptyset$ or $V_2 = \emptyset$, the unique player controlling all $V$ has a 
uniform optimal strategy in the game $\G^\sqsubseteq$. 
Then for all finite two-player games $\G^\sqsubseteq$ both players 
have uniform optimal strategies in the games $\G^\sqsubseteq$.  
\end{theorem}

This characterization will allow us to show the uniform-determinacy of lexicographic cost games with the $\MPInf$, $\MPSup$, $\LimInf$, and $\LimSup$ cost functions. The cases of $\Inf$, $\Sup$ and $\Disc^\lambda$ will be studied differently. 

We also need the next proposition about the cycle decomposition of an  infinite path in an arena $(V,E)$. Given a path $\rho$, we consider its \emph{cycle decomposition} into a multiset of simple cycles as follows.\footnote{We can similarly consider the cycle decomposition of finite paths $\rho$.} We push successively vertices $\rho _0, \rho_1, \dots$ onto a stack. Whenever we push a vertex $\rho_m$ equal to a vertex $\rho_n$ in the stack (i.e. a simple cycle $\rho_n \dots \rho_{m}$ is formed), we remove it from the stack (i.e. we remove all the vertices above $\rho_n$, but not $\rho_n$) and add it to the cycle decomposition multiset of $\rho$. Notice that at any moment, the stack contains the vertices of a simple path, thus at most $|V|$ vertices. 

\begin{proposition} \label{prop:decomposition}
Let $(V,E)$ be an arena. For each infinite path $\rho$ in $(V,E)$,  
there exists $n_0 \in \IN$ such that for all $n \geq n_0$, $(\rho_n, \rho_{n+1})$ 
is an edge of a cycle that appears in the cycle decomposition multiset of $\rho$. 
\end{proposition}

\begin{proof}
%
As $\rho$ is infinite, there exists $n_0 \in \IN$ such that 
for all $n \geq n_0$, each vertex $\rho_n$ is infinitely often repeated in $\rho_{\geq n_0}$. 
Consider the edge $(\rho_n,\rho_{n+1})$, with $n \geq n_0$. When $q = \rho_n$ is put 
on the stack, either a cycle is formed and removed from the stack, or a cycle is not 
yet been formed. In both cases, $q$ is at the top of the stack. Then $q' = \rho_{n+1}$ 
is pushed on the stack above $q$, as well as $\rho_{n+2}, \rho_{n+3}, \ldots$ and formed 
simple cycles are successively removed. At a certain point, a simple cycle $c$ containing 
the edge $(q,q')$ will be formed since $q$ appears infinitely often along $\rho$. 
\qed\end{proof}

Let us now recall some useful properties about \emph{reachability games} (resp. \emph{safety games}, \emph{co-B\"uchi} games, \emph{Rabin games})\cite{LNCS2500,EmersonJ88,PnueliR89}. They are two-player zero-sum games played on an arena which is a finite directed graph $(V,E)$ with no weights. A set $A \subseteq V$ of vertices is given in the case of reachability, safety and co-B\"uchi games, whereas a finite set of pairs $(A_k,B_k)$, with $A_k, B_k \subseteq V$ is given in the case of Rabin games. The two players play in a turn-based manner from a given initial vertex $v_0$. In the case of reachability (resp. safety) game, the produced play~$\rho$ is won by player~1 if $\rho$ visits some vertex of $A$ (resp. no vertex of $A$), and by player~2 otherwise. The set of initial vertices $v_0$ from which player 1 has a winning strategy can be computed in time $O(|V| + |E|)$, and a winning strategy can be chosen positional and computed in time $O(|V| + |E|)$~\cite{LNCS2500}. Moreover deciding whether player~1 has a winning strategy from an initial vertex $v_0$ is known to be {\sf P}-complete \cite{Immerman81}.
For co-B\"uchi games, play $\rho$ is won by player~1 if $\InfVertex(\rho) \cap A = \emptyset$ (with $\InfVertex(\rho)$ being the set of vertices infinitely visited along $\rho$), and by player 2 otherwise. For Rabin games, this condition is replaced by $\exists k: \InfVertex(\rho) \cap A_k = \emptyset \wedge \InfVertex(\rho) \cap B_k \neq \emptyset$. 
One can decide if player 1 has a winning strategy in time $\bigO{n^2}$ \cite{ChatterjeeH12} for co-B\"uchi games and in time $O(n^{p+1}\cdot p!)$ \cite{PitermanP06} for Rabin games, with $n$ the number of vertices and $p$ the number of pairs. Moreover for co-B\"uchi games, this problem is known to be {\sf P}-complete, as this is already the case for the AND-OR graph reachability problem~\cite{Immerman81}. 

We conclude this section with a remark that will be useful when studying the determinacy of lexicographic payoff games.

\begin{remark} \label{rem:positive} 
Without lost of generality we can assume that the game $\Glexun = (V, V_1, V_2, E, r, \g, \lexicoleq_1)$ uses a weight function $r = (r_1,r_2)$ such that $r_i : E \to \IN$ (instead of $r_i : E \to \IQ$) for each $i \in \set{1,2}$. Indeed all the weights appearing in $\Glexun$ have the form $\frac{a}{b}$ with $a \in \IZ$ and $b \in \IN$. 
Let $a^\ast$ be the smallest weight numerator ($0$ if weights are positive in $\Glexun$) and $b^\ast$ be the least common multiple of the weight denominators.
Then we define a new lexicographic payoff game $\G'^{\lexicoleq_1} = (V, V_1, V_2, E, r', \g, \lexicoleq_1)$ such that $r'_i(e) =  r_i(e) \cdot b^\ast  - a^\ast \cdot b^\ast$ for all $e \in E$ and $i \in \set{1,2}$. For the payoff functions $\g$ of Definition~\ref{def:payoff},
given a play $\rho$ in $\Glexun$, its payoff is multiplied by $b^\ast$ and shifted by the value $(-a^\ast \cdot b^\ast,-a^\ast \cdot b^\ast)$ when it is seen as a play in $\G'^{\lexicoleq_1}$. The same holds for values (if they exist), and positional optimal strategies correspond in both games. 
\end{remark}

\subsection{$\MPInf$ and $\MPSup$ Lexicographic Payoff Games} \label{subsec:MP}

In this section we prove that every $\MPInf$ lexicographic payoff game is uniformly-determined (see Theorem~\ref{thm:uniform-det}). We also give the proof of Theorem~\ref{thm:compl-lexico} for $\MPInf$ payoff.
Notice that we do not provide the proofs for $\MPSup$ payoff since they are similar.

Theorems~\ref{thm:uniform-det} and~\ref{thm:compl-lexico} for $\MPInf$ payoff are non trivial extensions to dimension~2 of well-known results about (one-dimensional) mean-payoff games \cite{EM79,ZP}. 
In \cite{KTB}, similar results have been proved for a class of zero-sum games called lexicographic mean-payoff games. It should be noted that this class is different from our class of $\MPInf$ lexicographic payoff games. Indeed, the authors use a lexicographic order $\lexicoleq'_1$ on $\IR^2$ different from ours:  $(x_1, x_2) \lexicoleq'_1 (x'_1, x'_2)$ iff $(x_1 < x'_1) \lor (x_1 = x'_1 \land x_2 \leq x'_2)$, and most importantly the payoff of a play $\rho$ is computed as $ \limAvgInf (r_1(\rho_k, \rho_{k+1}), r_2(\rho_k, \rho_{k+1}))$ at the level of two-dimensional vectors and with the $\liminf$ using order $\lexicoleq'_1$ (whereas we proceed componentwise). An example showing that the two definitions are not equivalent is given in \cite{KTB}\footnote{This example is given in the ArXiv version of \cite{KTB}, following discussions with one of the authors.}.

Let us first proceed to the proof of Theorem~\ref{thm:uniform-det} with $\MPInf$ payoff. 
For this purpose we use Theorem~\ref{thm:UniformlyDeterminedHugo}. We show that each 
$\MPInf$ lexicographic payoff game played by only one player is uniformly determined. 
At the end of this section, we will establish the related complexity stated in Theorem~\ref{thm:compl-lexico}.

\paragraph{Uniform-Determinacy.} 
Let $\Glexun = (V, V_1, V_2, E, r, \MPInf, \lexicoleq_1)$ be a lexicographic payoff game with a weight function with natural weights (see Remark~\ref{rem:positive}). 
For each $i \in \set{1, 2}$, we denote by $|r_i| = \max\{r_i(e) \mid e \in E\}$ the maximal weight for component $i$, and by $|R| = \max\set{|r_1|, |r_2|}$ the maximal weight of the game.

Recall that $\lexicoleq_1$ is a preference relation. Then, by Theorem~\ref{thm:UniformlyDeterminedHugo}, to show that $\Glexun$ is uniformly-determined, it suffices to show 
that if the vertices of the game are controlled by only one player, this player has a uniform optimal strategy. 

\begin{proposition}\label{prop:uniformlyDeterminedOnePlayerGame}
Let $\Glexun = (V, V_1, V_2, E, r, \MPInf, \lexicoleq_1)$ be an $\MPInf$ lexicographic payoff game. 
If $V_2 = \emptyset$ (resp. $V_1 = \emptyset$), then player $1$ (resp. player $2$) 
has a uniform optimal strategy in $\Glexun$. 
\end{proposition}

The proof of Proposition~\ref{prop:uniformlyDeterminedOnePlayerGame} is inspired from~\cite{GimbertZ05} and~\cite{KTB} with non trivial adaptations. It uses the following notation and lemma. Given a simple cycle $c$, we denote for each $i \in \set{1,2}$ 
\begin{equation*}
\M_i(c) = \frac{1}{n - m} \sum_{k = m}^{n-1} r_i(\rho_k, \rho_{k + 1}) 
\end{equation*}
the mean-payoff of the cycle $c$ according to component $i$. 

\begin{lemma}[see for instance \cite{LaurentDoyen}]\label{LimInfSup}
For all sequences $(a_n)_{n \in \IN}$ and $(b_n)_{n \in \IN}$ of reals 
numbers, we have:
\begin{eqnarray*}
\LSup (a_n + b_n) & \leq & \LSup (a_n) + \LSup (b_n) \\
\LSup (a_n - b_n) & \geq & \LSup (a_n) - \LSup (b_n) \\
\LInf (a_n + b_n) & \geq & \LInf (a_n) + \LInf (b_n) \\
\LInf (a_n - b_n) & \leq & \LInf (a_n) - \LInf (b_n). \\
\end{eqnarray*}
\end{lemma}

\begin{proof}[of Proposition~\ref{prop:uniformlyDeterminedOnePlayerGame}]
Let $\Glexun = (V, V_1, V_2, E, r, \MPInf, \lexicoleq_1)$ be an $\MPInf$ lexicographic payoff game. 

(a) Suppose first that $V_2 = \emptyset$. 
Let us denote by $\A = (V, V_1, V_2, E)$ the arena of $\Glexun$ where $V_2 = \emptyset$ and by $\Val^{\A}(v)$ the value of $v$ in $\A$. Remark that, in this case, we have that $\Val^{\A}(v) = \sup_{\sigma_1 \in \Sigma_1}\g(\out{\sigma_1}_v) = \sup_{\rho \in V^\omega, \rho_0 = v} \g(\rho)$. Let us construct a uniform optimal strategy $\sigma_1^\star$ for player~$1$ in $\A$ by induction on the size $|V|$ of $\A$ such that for all vertex $v \in V, \g(\out{\sigma_1^\star}_v) = \Val^{\A}(v)$. 

The initial case ($|V| = 1$) is trivial. Suppose that $|V| > 1$. Consider $c_{max}$ a simple cycle in $\A$ with the maximal mean-payoff $(m_1, m_2)$ w.r.t. the lexicographic order $\lexicoleq_1$. Let $W$ be the set of all the vertices of $V$ from which there is a path in $\A$ to reach $c_{max}$. Let us define the strategy $\sigma_1^\star$ on $W$, such that $\sigma_1^\star$ goes as quickly as possible to this maximum payoff cycle and next go round this cycle forever. The produced play $\rho^\star$ from a vertex of $W$ has then payoff $\MPInf(\rho^\star) = \MPInf(c_{max}^\omega) = (\M_1(c_{max}), \M_2(c_{max})) = (m_1, m_2)$.  

Let us now show that for all vertices $v \in W$ any infinite play from $v$ in $\A$ cannot supply a payoff 
greater than $(m_1, m_2)$. Let $\rho$ be any play in $\A$ from $v$. 
Let $\rho_{\leq n}$ be any prefix of $\rho$. By definition of $c_{max}$, each simple cycle $c$ 
of the cycle decomposition multiset of $\rho_{\leq n}$ satisfies 
\begin{align}
\mbox{either} & \quad \M_1(c) = m_1   \mbox{ and }  \M_{2}(c) \geq m_2\label{equation:cas1} \\
\mbox{or}  & \quad \M_1(c) \leq \bar m_1 < m_1.  \label{equation:cas2}
\end{align}

For each prefix $\rho_{\leq n}$ of $\rho$ of length $n$, let us denote 
by $J_1(n)$ (resp. $J_2(n)$) the sum of the lengths of all simple cycles in its cycle decomposition 
that satisfy Property ($\ref{equation:cas1}$) (resp. Property ($\ref{equation:cas2}$)). Recall that $n - J_1(n) - J_2(n) \leq |V|$ by definition of the cycle decomposition.
First observe that for $i \in \set{1, 2}$, 
\begin{equation*}\label{LI}
  0 \leq \LInf\left(\frac{J_i(n)}{n}\right) \leq \LSup\left(\frac{J_i(n)}{n}\right) \leq 1. 
\end{equation*} 
and that 
\begin{equation*}
  \lim_{n \to \infty} \frac{J_1(n) + J_2(n)}{n} = 1, 
\end{equation*}
since
\begin{equation*}
  1 = \lim_{n \to \infty} \frac{n}{n} \geq \lim_{n \to \infty} \frac{J_1(n) + J_2(n)}{n} \geq \lim_{n \to \infty} \frac{n - |V|}{n} = 1.
\end{equation*}

For any $n \in \IN$ we have 
\begin{eqnarray*}
  \sum_{k = 0}^{n - 1} r_1(\rho_k, \rho_{k + 1}) & \leq & J_1(n) \cdot m_1 + J_2(n) \cdot \bar m_1 + |V| \cdot |R|, \textrm{ and } \label{ineq1} \\
  & & \nonumber \\
  \sum_{k = 0}^{n - 1} r_2(\rho_k, \rho_{k + 1}) & \geq & J_1(n)\cdot m_2. \label{ineq2} 
\end{eqnarray*}
   
It follows that
\begin{eqnarray*}
  \MPInf_{1}(\rho) & \leq &  \LInf\left(\frac{J_1(n)}{n} (m_1 - \bar m_1 + \bar m_1) + \frac{J_2(n)}{n} \cdot \bar m_1  + \frac{|V|}{n} \cdot |R|\right) \\
  & & \\
  & \leq & \LInf\left(\frac{J_1(n)}{n} \cdot (m_1 - \bar m_1 ) \right) \\ 
  & & - \LInf\left(- \frac{J_1(n) + J_2(n)}{n} \cdot \bar m_1 - \frac{|V|}{n} \cdot |R| \right) \\
  & & \\
  & \leq & \LInf\left(\frac{J_1(n)}{n}  \cdot (m_1 - \bar m_1 )  \right) + \LSup\left(\frac{J_1(n) + J_2(n)}{n} \cdot \bar m_1 \right) \\
  & & \\
  & = & \LInf\left(\frac{J_1(n)}{n}\right) \cdot (m_1 - \bar m_1 ) + \bar m_1.
\end{eqnarray*}

Recall that $\mu = \LInf\left(\frac{J_1(n)}{n}\right) \leq 1$.

\begin{enumerate}[(a)]
  
\item If $\mu < 1$, then $\MPInf_{1}(\rho) < m_1$.

\item If $\mu = 1$, then  $\MPInf_{1}(\rho) \leq m_1$, and 
  \begin{equation*}
    \MPInf_{2}(\rho) \geq \LInf\left(\frac{J_1(n)}{n}.m_2\right) = m_2.
  \end{equation*}

\end{enumerate}

We can conclude that $\MPInf(\rho) \lexicoleq_1 (m_1, m_2)$. 


We then have that $\g(\out{\sigma_1^\star}_v) = \sup_{\rho \in V^\omega, \rho_0 = v} \g(\rho) = \Val^{\A}(v) = (m_1, m_2)$. If $|V| = |W|$, $\sigma_1^\star$ is the required uniform optimal strategy for player~$1$. Else, consider the subarena $\A' = {\A}_{\upharpoonright{V'}}$ of $\A$ restricted to the set of vertices $V' = V \setminus W$. By induction, we have  constructed a uniform optimal strategy ${\sigma'}_1^\star$ for player~$1$ in $\A'$ such that for all $v \in V', \g(\out{{\sigma'}_1^\star}_v) = \sup_{\sigma_1 \in \Sigma_1} \g(\out{\sigma_1}_v) = \Val^{\A'}(v)$. 
As all the plays $\rho$ from $v \in V \setminus W$ stay in $\A'$ by definition of the set $W$, we have that $\Val^{\A'}(v) = \Val^{\A}(v)$. The required uniform optimal strategy is then the union of $\sigma_1^\star$ (defined on $W$) and ${\sigma'}_1^\star$ (defined on $V \setminus W$). 

(b) Let us now suppose that $V_1 = \emptyset$ and show that player $2$ 
has a uniform optimal strategy in $\A$. 
This strategy is constructed by induction on the size $|V|$ of $\A$ as in the previous case except that, in the induction step (when $|V| > 1$), 
we choose a simple cycle $c_{min}$ with the minimal (instead of the maximal) mean-payoff $(m_1, m_2)$ w.r.t. the lexicographic order $\lexicoleq_1$.
Let $W$ be the set of all the vertices of $V$ from which there is a path in $\A$ to reach $c_{min}$. We define the strategy $\sigma_1^\star$ on $W$, such that $\sigma_1^\star$ goes as quickly as possible to this minimum payoff cycle and next go round this cycle forever. The produced play $\rho^\star$ from a vertex of $W$ has then cost $\MPInf(\rho^\star) = \MPInf(c_{min}^\omega) = (\M_1(c_{min}), \M_2(c_{min})) = (m_1, m_2)$. 
Let us now show that for all vertices $v \in W$ any infinite play from $v$ in $\A$ cannot supply a payoff smaller than $(m_1, m_2)$. Let $\rho$ be any play in $\A$. 
Let $\rho_{\leq n}$ be any prefix of $\rho$. 
By definition of $c_{min}$, for each simple cycle $c$ of the cycle decomposition multiset of $\rho_{\leq n}$ 
Properties ($\ref{equation:cas1}$) and ($\ref{equation:cas2}$) are now replaced by the next ones:  
\begin{align}
\mbox{either} & \quad \M_1(c) = m_1   \mbox{ and }  \M_{2}(c) \leq m_2\label{equation:cas3} \\
\mbox{or}  & \quad \M_1(c) \geq \bar m_1 > m_1. \label{equation:cas4}
\end{align}

For each prefix $\rho_{\leq n}$ of $\rho$ of length $n$, let us denote 
by $J_1(n)$ (resp. $J_2(n)$) the sum of the lengths of all simple cycles in its cycle decomposition 
that satisfy Property ($\ref{equation:cas3}$) (resp. Property ($\ref{equation:cas4}$)). 

For all $n \in \IN$ we have 
\begin{eqnarray*}
  \sum_{k = 0}^{n - 1} r_1(\rho_k, \rho_{k + 1}) & \geq & J_1(n)\cdot m_1 + J_2(n)\cdot \bar m_1 , \textrm{ and } \\
  & & \\
  \sum_{k = 0}^{n - 1} r_2(\rho_k, \rho_{k + 1}) & \leq & J_1(n)\cdot m_2 + J_2(n)\cdot |R| + |V|\cdot |R|. 
\end{eqnarray*}
Let $\kappa = \LInf \frac{J_2(n)}{n}$. We have either $\kappa > 0$ or $\kappa = 0$. We have (notice the usage of Lemma~\ref{LimInfSup}):
\begin{eqnarray*}
  \MPInf_{1}(\rho) & \geq & \LInf\left(\frac{J_1(n)}{n} \cdot m_1 + \frac{J_2(n)}{n} \cdot \bar m_1 \right) \\
  & & \\
  & \geq & \LInf\left(\frac{J_1(n) + J_2(n)}{n} \cdot m_1\right) + \LInf\left(\frac{J_2(n)}{n} \cdot (\bar m_1 - m_1) \right) \\
  & & \\
  & = & m_1 + \kappa \cdot (\bar m_1 - m_1).
\end{eqnarray*}
We can now prove that $(m_1, m_2) \lexicoleq_1 \MPInf(\rho)$. 
\begin{enumerate}[(a)]
\item If $\kappa > 0$, since $\bar m_1 > m_1$, we have 
  $$\MPInf_{1}(\rho) \geq  m_1 + \kappa \cdot (\bar m_1 - m_1) > m_1$$ 
  and so $(m_1, m_2) \lexicoleq_1 \MPInf(\rho)$.
\item If $\kappa = 0$, then 
  $$\MPInf_{1}(\rho) \geq m_1,$$ 
  and (using again Lemma~\ref{LimInfSup} and recalling that $|R|, m_2 \geq 0$)
  \begin{eqnarray*}
    \MPInf_{2}(\rho) & \leq & \LInf\left(\frac{J_1(n)}{n} \cdot m_2 + \frac{J_2(n)}{n} \cdot |R| + \frac{|V|}{n} \cdot |R|\right) \\
    & & \\
    & \leq & \LInf\left(\frac{J_2(n)}{n} |R|\right) - \LInf\left(\frac{-J_1(n)}{n} \cdot m_2 + \frac{-|V|}{n} \cdot |R| \right) \\
    & & \\
    & \leq & \LInf\left(\frac{J_2(n)}{n}\right) \cdot |R| + \LSup\left(\frac{J_1(n)}{n}\right)\cdot m_2 \\ 
	& & + \LSup \left(\frac{|V|}{n} \cdot |R| \right) \\
    & & \\
    & \leq & 0 \cdot |R|  + 1 \cdot m_2 + 0 = m_2.
  \end{eqnarray*}
  In this case we thus also have $(m_1, m_2) \lexicoleq_1 \MPInf(\rho)$.    
\end{enumerate}

We can conclude the proof with similar arguments as in case (a). 	
\qed\end{proof}

Theorem~\ref{thm:uniform-det} for $\MPInf$ follows from Theorem~\ref{thm:UniformlyDeterminedHugo} 
and Proposition~\ref{prop:uniformlyDeterminedOnePlayerGame}.

\paragraph{Complexity Results.}
We have proved that every $\MPInf$ lexicographic payoff game is uniformly-determined. Let us now established the related complexities as given in Theorem~\ref{thm:compl-lexico}.

\begin{proof}[of Theorem~\ref{thm:compl-lexico} for $\MPInf$ payoff] 
The proof is inspired by a proof given in \cite{KTB}: we are going to sketch a reduction of $\MPInf$ lexicographic payoff games to mean-payoff games for optimal strategies. 
This reduction keeps the same arena $(V,E)$ but replaces the weight function $(r_1, r_2)$ by a single weight function $r^\ast$. We again suppose that the weights are positive integers by Remark~\ref{rem:positive}. We know by Theorem~\ref{thm:uniform-det} that $\Glexun$ is uniformly-determined. We can thus restrict our attention to simple cycles. The mean-payoff of a simple cycle $C$ in $\Glexun$ is of the form $(\frac{a}{n},\frac{b}{n})$ such that each component is a bounded rational (between $0$ and $|r_1|$ (resp. $|r_2|$)), and the denominator is at most equal to $|V|$. We define $m = |V|^2 \cdot |r_2| + 1$ and $r^\ast = r_1 \cdot m - r_2$. It follows that in the resulting mean-payoff game, the cycle $C$ has a mean-payoff equal to $\frac{a \cdot m - b}{n}$. Notice that given two simple cycles $C_1$ and $C_2$ in $\Glexun$, if the first components of their mean-payoff are distinct, that is, $|\frac{a_1}{n_1} - \frac{a_2}{n_2}| > 0$, then $|\frac{a_1}{n_1} - \frac{a_2}{n_2}| \geq \frac{1}{|V|^2}$ because the weights are integers. By the choice of $m$, let us show that 
\begin{eqnarray} \label{eq:reductionMP}
\frac{a_1}{n_1} < \frac{a_2}{n_2} \quad\Rightarrow\quad \frac{a_1 \cdot m - b_1}{n_1} < \frac{a_2 \cdot m - b_2}{n_2}.
\end{eqnarray}
Indeed, we have 
\begin{eqnarray*}
\frac{a_1}{n_1} + \frac{1}{|V|^2} &\leq& \frac{a_2}{n_2} \\
\frac{a_1}{n_1}\cdot m + |r_2| + \frac{1}{|V|^2}   &\leq&   \frac{a_2}{n_2}\cdot m \\
\frac{a_1 \cdot m - b_1}{n_1} \leq \frac{a_1}{n_1}\cdot m &<&   \frac{a_2}{n_2} \cdot m -  |r_2| \leq \frac{a_2 \cdot m - b_2}{n_2}. \\
\end{eqnarray*}
To prove the correctness of the proposed reduction, it is enough to prove that  $(\frac{a_1}{n_1},\frac{b_1}{n_1}) \lexicoleq_1 (\frac{a_2}{n_2},\frac{b_2}{n_2})$ in $\Glexun$ if and only if $\frac{a_1 \cdot m - b_1}{n_1} \leq \frac{a_2 \cdot m - b_2}{n_2}$ in the resulting mean-payoff game (the comparaison of two cycle payoffs coincide in both games). Suppose first that $(\frac{a_1}{n_1},\frac{b_1}{n_1}) \lexicoleq_1 (\frac{a_2}{n_2},\frac{b_2}{n_2})$, i.e., either $\frac{a_1}{n_1} < \frac{a_2}{n_2}$, or ($\frac{a_1}{n_1} = \frac{a_2}{n_2}$ and $\frac{b_1}{n_1} \geq \frac{b_2}{n_2}$). In both cases we get $\frac{a_1 \cdot m - b_1}{n_1} \leq \frac{a_2 \cdot m - b_2}{n_2}$ (we use ($\ref{eq:reductionMP}$) in the first case). Suppose now that $\frac{a_1 \cdot m - b_1}{n_1} \leq \frac{a_2 \cdot m - b_2}{n_2}$. It follows that $\frac{a_1}{n_1} \leq \frac{a_2}{n_2}$ by ($\ref{eq:reductionMP}$). Moreover if $\frac{a_1}{n_1} = \frac{a_2}{n_2}$, then $\frac{b_1}{n_1} \geq \frac{b_2}{n_2}$. Therefore $(\frac{a_1}{n_1},\frac{b_1}{n_1}) \lexicoleq_1 (\frac{a_2}{n_2},\frac{b_2}{n_2})$.

By this reduction and classical results on mean-payoff games~\cite{EM79,ZP}, we get the three statements of Theorem~\ref{thm:compl-lexico}.
\qed \end{proof}

\subsection{$\LimInf$ and $\LimSup$ Lexicographic Payoff Games}

In this section, we consider $\LimInf$ lexicographic payoff games and we provide proofs of Theorems~\ref{thm:uniform-det} and~\ref{thm:compl-lexico} for these games.
The proofs are not given for $\LimSup$ payoff because they are simple adaptations of the ones given for  $\LimInf$ payoff. 

\paragraph{Uniform-Determinacy.} 

The proof of Theorem~\ref{thm:uniform-det} for $\LimInf$ payoff has the same structure 
as for $\MPInf$ payoff. Let $\Glexun$ be a $\LimInf$ lexicographic payoff game. 
Without lost of generality, we can suppose the weights used in $\Glexun$ are natural 
numbers (see Remark~\ref{rem:positive}). 
As for lexicographic payoff games with $\MPInf$ payoff function, 
we show (Proposition~\ref{prop:uniformlyDeterminedOnePlayerGameLimInf}) 
that if the vertices of the game are controlled by only one player, this player has a uniform optimal strategy. Theorem~\ref{thm:uniform-det} will then follow by Theorem~\ref{thm:UniformlyDeterminedHugo}. 

\begin{proposition}\label{prop:uniformlyDeterminedOnePlayerGameLimInf}
Let $\Glexun = (V, V_1, V_2, E, r, \LimInf, \lexicoleq_1)$ be a $\LimInf$ lexicographic payoff game. 
If $V_2 = \emptyset$ (resp. $V_1 = \emptyset$), then player $1$ (resp. player $2$) 
has a uniform optimal strategy. 
\end{proposition}

Given a simple cycle $c$, we denote for each $i \in \set{1,2}$
\begin{equation*}
\Min_i(c) = \min_{m \leq k < n} r_i(\rho_k, \rho_{k + 1}) 
\end{equation*}
the minimum of the cycle according to component $i$. 

\begin{proof}[of Proposition~\ref{prop:uniformlyDeterminedOnePlayerGameLimInf}]
Let $\Glexun = (V, V_1, V_2, E, r, \LimInf, \lexicoleq_1)$ be a $\LimInf$ lexicographic poayoff game. 

(a) Suppose that $V_2 = \emptyset$. We then show that player $1$ 
has a uniform optimal strategy.
This strategy is constructed as in the proof of Proposition~\ref{prop:uniformlyDeterminedOnePlayerGameLimInf}. 
Let us denote by $\A = (V, V_1, V_2, E)$ 
the arena of $\Glexun$ where $V_2 = \emptyset$ and by $\Val^{\A}(v)$ the value of $v$ in $\A$. 
We have that $\Val^{\A}(v) = \sup_{\sigma_1 \in \Sigma_1}\g(\out{\sigma_1}_v) = \sup_{\rho \in V^\omega, \rho_0 = v} \g(\rho)$. Let us construct a uniform optimal strategy $\sigma_1^\star$ for player~$1$ in $\A$ by induction on the size $|V|$ of $\A$ such that for all $v \in V, \g(\out{\sigma_1^\star}_v) = \Val^{\A}(v)$. 

The initial case ($|V| = 1$) is trivial. Suppose that $|V| > 1$. Consider $c_{max}$ a simple cycle in $\A$ with the maximal $(\Min_1, \Min_2)$ payoff $(m_1, m_2)$ w.r.t. the lexicographic order $\lexicoleq_1$. Let $W$ be the set of all the vertices of $V$ from which there is a path in $\A$ to reach $c_{max}$. Let us define the strategy $\sigma_1^\star$ on $W$, such that $\sigma_1^\star$ goes as quickly as possible to this maximum payoff cycle and next go round this cycle forever. The produced play $\rho^\star$ from a vertex of $W$ has then cost $\LimInf(\rho^\star) = \LimInf(c_{max}^\omega) = (\Min_1(c_{max}), \Min_2(c_{max})) = (m_1, m_2)$.  

Let us now show that for all vertices $v \in W$ any infinite play from $v$ in $\A$ cannot supply a cost greater than $(m_1, m_2)$. Let $\rho$ be any play in $\A$ from $v$. 
Let $\rho_{\leq n}$ be any prefix of $\rho$. By definition of $c_{max}$, each simple cycle $c$ of the cycle decomposition multiset of $\rho_{\leq n}$ satisfies 
\begin{align}
\mbox{either} & \quad \Min_1(c) = m_1   \mbox{ and }  \Min_{2}(c) \geq m_2\label{eq:cas5} \\
\mbox{or}  & \quad \Min_1(c) \leq \bar m_1 < m_1. \label{eq:cas6}
\end{align}

By Proposition \ref{prop:decomposition}, there is an index $n_0 \in \IN$ 
such that each edge $(\rho_n, \rho_{n+1})$ with $n \geq n_0$, is in a simple cycle $c_n$ 
that appears in the cycle decomposition multiset of $\rho$. 
Then, for $n \geq n_0$, we define 
    $$
    u_n = \left\{
      \begin{array}{ll}
        1 & \mbox{if cycle $c_n$ satisfies Property (\ref{eq:cas5})}  \\
        0 & \mbox{if cycle $c_n$ satisfies Property (\ref{eq:cas6})}. 
      \end{array}
    \right.
    $$
Remark that $\liminf_{n \to \infty} u_n = 0$ or $1$.

If $\liminf_{n \to \infty} u_n = 1$, then there exists some $n_1 \geq n_0$ such that all 
cycles $c_n$, for $n \geq n_1$, satisfy $\Min_1(c_n) = m_1$ and $\Min_2(c_n) \geq m_2$. 
Therefore $\LimInf_1(\rho) = m_1$ and $\LimInf_2(\rho) \geq m_2$.

If $\liminf_{n \to \infty} u_n = 0$, there are infinitely many cycles $c_n$, with $n \geq n_0$ such that $\Min_1(c_n) \leq \bar m_1 < m_1$. Therefore $\LimInf_1(\rho) < m_1$.

It follows that in both cases $\LimInf(\rho) \lexicoleq_1 (m_1, m_2)$.

We then have that $\g(\out{\sigma_1^\star}_v) = \sup_{\rho \in V^\omega, \rho_0 = v} \g(\rho) = \Val^{\A}(v) = (m_1, m_2)$. If $|V| = |W|$, $\sigma_1^\star$ is the required uniform optimal strategy for player~$1$. Else, consider the subarena $\A' = {\A}_{\upharpoonright{V'}}$ of $\A$ restricted to the set of vertices $V' = V \setminus W$. By induction, we have constructed a uniform optimal strategy ${\sigma'}_1^\star$ for player~$1$ in $\A'$ such that for all $v \in V', \g(\out{{\sigma'}_1^\star}_v) = \sup_{\sigma_1 \in \Sigma_1} \g(\out{\sigma_1}_v) = \Val^{\A'}(v)$ 
As all the plays $\rho$ from $v \in V \setminus W$ stay in $\A'$ by definition of the set $W$, we have that $\Val^{\A'}(v) = \Val^{\A}(v)$. The required uniform optimal strategy is then the union of $\sigma_1^\star$ (defined on $W$) and ${\sigma'}_1^\star$ (defined on $V \setminus W$).

(b) Let us now suppose that $V_1 = \emptyset$ and show that player $2$ 
has a uniform optimal strategy in $\A$. 
This strategy is constructed by induction on the size $|V|$ of $\A$ as in the previous case except that, in the induction step (when $|V| > 1$),   
we choose a simple cycle $c_{min}$ with the minimal (instead of the maximal) $(\Min_1, \Min_2)$ payoff $(m_1, m_2)$ w.r.t. the lexicographic order $\lexicoleq_1$.
Let $W$ be the set of all the vertices of $V$ from which there is a path in $\A$ to reach $c_{min}$. We define the strategy $\sigma_1^\star$ on $W$, such that $\sigma_1^\star$ goes as quickly as possible to this minimum payoff cycle and next go round this cycle forever. The produced play $\rho^\star$ from a vertex of $W$ has then payoff $\LimInf(\rho^\star) = \LimInf(c_{min}^\omega) = (\Min_1(c_{min}), \Min_2(c_{min})) = (m_1, m_2)$. 
Let us now show that for all vertices $v \in W$ any infinite play from $v$ in $\A$ cannot supply a payoff smaller than $(m_1, m_2)$. Let $\rho$ be any play in $\A$. 
Let $\rho_{\leq n}$ be any prefix of $\rho$. 
By definition of $c_{min}$, for each simple cycle $c$ of the cycle decomposition multiset of $\rho_{\leq n}$ Properties~\ref{eq:cas5} and ~\ref{eq:cas6} are now replaced by the next ones: 
\begin{align}
\mbox{either} & \quad \Min_1(c) = m_1   \mbox{ and }  \Min_{2}(c) \leq m_2\label{eq:cas7} \\
\mbox{or}  & \quad \Min_1(c) \geq \bar m_1 > m_1. \label{eq:cas8}
\end{align}
We define
   $$
    u_n = \left\{
      \begin{array}{ll}
        0 & \mbox{if cycle $c_n$ satisfies Property (\ref{eq:cas7})}  \\
        1 & \mbox{if cycle $c_n$ satisfies Property (\ref{eq:cas8})}. 
      \end{array}
    \right.
    $$
Then, either $\liminf_{n \to \infty} u_n = 0$ and $\LimInf_1(\rho) = m_1$, $\LimInf_2(\rho) \leq m_2$, or $\liminf_{n \to \infty} u_n = 1$ and $\LimInf_1(\rho) > m_1$. 
We conclude that $(m_1, m_2) \lexicoleq_1 \LimInf(\rho)$. We can conclude the proof with similar arguments as in case (a). 
\qed\end{proof}

The proof of Theorem~\ref{thm:uniform-det} for $\LimInf$ cost follows from 
Theorem~\ref{thm:UniformlyDeterminedHugo} and 
Proposition~\ref{prop:uniformlyDeterminedOnePlayerGameLimInf}.

\paragraph{Complexity Results.}

Let us now turn to complexity results and give a proof of Theorem~\ref{thm:compl-lexico} for $\LimInf$ payoff.
To get the announced polynomial complexities, the previous reduction to finite cycle-forming games is too expensive. We here propose another reduction to co-B\"uchi and (one pair) Rabin games.

We have the next lemma. The proof of Theorem~\ref{thm:compl-lexico} will follow. We suppose again that weights are natural numbers by Remark~\ref{rem:positive}.

\begin{lemma} \label{lem:reduction}
Let $\Glexun$ be a $\LimInf$ lexicographic payoff game, $v \in V$ be a vertex and $(\alpha,\beta) \in \IN^2$ be a pair of naturals. Deciding whether player 1 has a strategy $\sigma_1$ from $v$ such that $(\alpha,\beta) \lexicoleq_1 \LimInf(\out{\sigma_1, \sigma_2}_v)$ for all strategies $\sigma_2$ of player 2 is {\sf P}-complete. 
\end{lemma}

\begin{proof}
We first prove that one can decide in $O((|V| + |E|)^2)$ whether player 1 has a strategy $\sigma_1$ from $v$ such that $(\alpha,\beta) \lexicoleq_1 \LimInf(\out{\sigma_1, \sigma_2}_v)$ for all strategies $\sigma_2$ of player 2.
Let us fix $v \in V$, $(\alpha,\beta) \in \IN^2$, and a strategy $\sigma_1$ of player 1 in $\Glexun$. Consider a strategy $\sigma_2$ of player 2 and the outcome $\rho = \out{\sigma_1, \sigma_2}_v$. We have that $(\alpha,\beta) \lexicoleq_1 \LimInf(\rho)$ iff 
\begin{itemize}
\item there exists $n \geq 0$ such that  $r_1(\rho_k,\rho_{k+1}) \geq \alpha + 1$ for all $k \geq n$, or
\item there exists $n \geq 0$ such that  $r_1(\rho_k,\rho_{k+1}) \geq \alpha$ for all $k \geq n$ and $r_2(\rho_l,\rho_{l+1}) \leq \beta$ for infinitely many $l$. 
\end{itemize}
Using LTL notation, we abbreviate these conditions by 
\begin{eqnarray} \label{eq:rabin}
\rho \models (\Diamond \Box~ r_1 \geq \alpha + 1) \quad \vee \quad \rho \models (\Diamond \Box~ r_1 \geq \alpha) \wedge (\Box\Diamond~ r_2 \leq \beta).
\end{eqnarray}
The first condition is a co-B\"uchi condition, while the second one is a Rabin condition. These conditions can then be encoded with two Rabin pairs. 

Let us construct the following Rabin game in relation with these conditions. 
Firstly, from the arena $(V,E)$ of $\Glexun$, we construct a new arena $(V',E')$ in a way to have weights depending on vertices instead of edges. We proceed as follows:  each edge $e = (v,v') \in E$ is split into two consecutive edges. The new intermediate vertex belongs to player 1\footnote{It could belong to player 2 since there is exactly one outgoing edge.}, and it  is decorated with $(r_1(e), r_2(e))$. The vertices of $V$ are decorated with $(+\infty, +\infty)$, and $V'$ has $n' = |V| + |E|$ vertices. 
Secondly, we define two Rabin pairs $(A_1,B_1)$ and $(A_2, B_2)$ such that $A_1 \subseteq V'$ (resp. $B_1 \subseteq V'$) is composed of all vertices decorated by a pair $(a,b)$ such that $a < \alpha + 1$ (resp. $b = + \infty$) and $A_2 \subseteq V'$ (resp. $B_2 \subseteq V'$) is composed of all vertices decorated by a pair $(a,b)$ such that $a < \alpha$ (resp. $b \leq \beta$). For a play $\rho^R$ in the constructed Rabin game and its corresponding play $\rho$ in $\Glexun$ we have that 
\begin{eqnarray*}
\rho^R \models (A_1, B_1) & \iff & \rho \models (\Diamond \Box~ r_1 \geq \alpha + 1), \mathrm{and } \\
\rho^R \models (A_2, B_2) & \iff & \rho \models (\Diamond \Box~ r_1 \geq \alpha) \wedge (\Box\Diamond~ r_2 \leq \beta). \\
\end{eqnarray*}
It follows that player 1 has a winning strategy in the constructed Rabin game from vertex $v$ iff player 1 has a strategy $\sigma_1$ in $\Glexun$. 
This property can be checked in time $O(n'^{p+1}\cdot p!)$ with $p=1$ (there is one Rabin pair). 
One can thus decide in time $O((|V| + |E|)^2)$ whether player 1 has a strategy $\sigma_1$ from $v$ such that $(\alpha,\beta) \lexicoleq_1 \LimInf(\out{\sigma_1, \sigma_2}_v)$ for all strategies $\sigma_2$ of player 2. 

To conclude the proof of Lemma~\ref{lem:reduction}, we show that for each co-B\"uchi game $\CB$, we can construct a $\LimInf$ lexicographic payoff game $\Glexun_{\CB}$ and a pair of naturals $(\alpha, \beta)$ such that the problem of deciding if player $1$ has a winning strategy from a vertex $v$ in $\CB$ is equivalent to the problem stated in Lemma~\ref{lem:reduction}. Let $(V,E)$ be the arena of $\CB$ and $A \subseteq V$ be the set given for the co-B\"uchi condition. We set $(\alpha, \beta) = (1, 1)$ and we define $\Glexun_{\CB} = (V, V_1, V_2, E, r, \LimInf)$ with the reward function $r$ such that for all $(v, v') \in E$, $r(v, v') = (0, 1)$ if $v' \in A$, and $r(v, v') = (1, 1)$ otherwise. Clearly, for every play $\rho$ starting in $v$, we have $\Inf(\rho) \cap A = \emptyset$ iff $\LimInf(\rho) = (1,1)$.  This completes the proof.
\qed\end{proof}

\begin{proof}[of Theorem~\ref{thm:compl-lexico} for $\LimInf$ payoff]
The first part is a direct consequence of Lemma~\ref{lem:reduction} by definition of the value.

The second is also obtained as a corollary of Lemma~\ref{lem:reduction}. Indeed, given a vertex $v$, it suffices to apply it to all pairs $(\alpha,\beta)$ such that $\alpha = r_1(e)$ and $\beta = r_2(e')$ for some edges $e, e' \in E$. Notice that there are $|E|^2$ pairs. The value of $v$ is thus the maximum among those pairs $(\alpha,\beta)$ for which player 1 has a strategy $\sigma_1$ from $v$ such that $(\alpha,\beta) \lexicoleq_1 \LimInf(\out{\sigma_1, \sigma_2}_v)$ for all strategies $\sigma_2$ of player 2. This argument shows that computing the value of $v$ can be done in polynomial time. 

The third statement is a consequence of the second one, using a dichotomy on the edges of $E$ as proposed in~\cite{ZP}. We first start by computing the values $\Val(v)$ for each $v$ in $\Glexun$. If all the vertices of $V$ have outdegree one, then player 1 has a unique uniform optimal strategy. Otherwise, let $v$ be a vertex with outdegree $d > 1$. Let us remove $\lceil \frac{d}{2} \rceil$ of the edges leaving $v$, and let us compute the value $\Val'(v)$ of $v$ in the resulting graph. If $\Val'(v) = \Val(v)$ (resp. $\Val'(v) \ne \Val(v)$), then there is a uniform optimal strategy for player 1 which does not use any of the removed edges (uses one of the removed edges). In both cases, we can restrict the computation to a subgraph of $\Glexun$ with at least $\lfloor \frac{d}{2} \rfloor$ fewer edges. After a polynomial\footnote{$O(\Sigma_{v \in V_1} \log d(v))$ with $d(v)$ being the outdegree of vertex $v$ in $\Glexun$.} number of such experiments, we get the required uniform optimal strategy for player 1 with the announced complexity. A uniform optimal strategy for player 2 can be found in the same way.
\qed\end{proof}

\subsection{$\Inf$ and $\Sup$ Lexicographic Payoff Games}

In this section, we prove Theorems~\ref{thm:uniform-det} and~\ref{thm:compl-lexico} for $\Inf$ and $\Sup$ lexicographic payoff games. Without lost of generality, we limit the proofs to the $\Inf$ payoff.

\paragraph{Positional-Determinacy.}

We begin by proving that these games  are positionally-determined. Some preliminary lemmas are necessary.

\begin{lemma} \label{lem:valueInf}
Every $\Inf$ lexicographic payoff game $\Glexun$ is determined, and its values can be computed in polynomial time.
\end{lemma}

\begin{proof}
Given $\Glexun$ and an initial vertex $v_0$, we derive a $\LimInf$ lexicographic payoff game $\G'^{\lexicoleq_1}$ and an initial vertex $v'_0 = (v_0, +\infty, +\infty)$, with the same construction as in the proof of Theorem~\ref{thm:existenceES} for $\Inf$ payoff (we augment the vertices with the current smallest seen weights). The latter game is uniformly-determined by Theorem~\ref{thm:uniform-det} and its values can be computed in polynomial by Theorem~\ref{thm:compl-lexico} (for $\LimInf$ payoff). In $\G'^{\lexicoleq_1}$, let $(\alpha,\beta)$ be the value of $v'_0$ and $\sigma'^{\star}_1$, $\sigma'^{\star}_2$ be a pair of uniform optimal strategies for player~1 and player~2 respectively. We naturally derive finite-memory strategies $\sigma^{\star}_1$, $\sigma^{\star}_2$ in $\Glexun$ that mimic strategies $\sigma'^{\star}_1$, $\sigma'^{\star}_2$. Let us show that in $\Glexun$, $v_0$ has value $(\alpha,\beta)$ and $\sigma^{\star}_1$, $\sigma^{\star}_2$ are optimal strategies\footnote{In the proof of Theorem~\ref{thm:uniform-det} for $\Inf$ payoff, we will explain how to obtain positional optimal strategies.}. In $\Glexun$, let $\sigma_2$ be any strategy for player~2, and let us prove that $(\alpha,\beta) \lexicoleq_1  \Inf(\out{\sigma^\star_1, \sigma_{2}}_{v_0})$ (see Lemma~\ref{lemma:pointDeSelle}). This inequality follows from $(\alpha,\beta) \lexicoleq_1  \LimInf(\out{\sigma'^\star_1, \sigma'_{2}}_{v'_0})$ in $\G'^{\lexicoleq_1}$ and $\Inf(\out{\sigma^\star_1, \sigma_{2}}_{v_0}) = \LimInf(\out{\sigma'^\star_1, \sigma'_{2}}_{v'_0})$, where $\sigma'_{2}$ is the strategy that mimics $\sigma_{2}$ in $\G'^{\lexicoleq_1}$. Similarly we have $\inf(\out{\sigma_1, \sigma^\star_{2}}_{v_0}) \lexicoleq_1 (\alpha,\beta)$ for all strategy $\sigma_1$ of player~1 in $\Glexun$. Therefore from this reduction to $\LimInf$ lexicographic payoff game, $\Glexun$ is determined, and its values can be computed in polynomial time.
\qed\end{proof}

\begin{lemma} \label{lem:partition}
Let $(\alpha, \beta)$ be a pair of rational numbers. 
\begin{enumerate}
\item There exists a partition of $V$ into $W_1$ and $W_2$ such that $v \in W_1$ iff player 1 has a (positional) strategy $\sigma_1^\ast$ from $v$ in $\Glexun$ such that $(\alpha,\beta) \lexicoleq_1 \Inf(\out{\sigma_1^\ast, \sigma_2}_{v})$ for all strategies $\sigma_2$ of player 2.
\item Similarly, there exists a partition of $V$ into $T_1$ and $T_2$ such that $v \in T_2$ iff player 2 has a (positional) strategy $\sigma_2^\ast$ from $v$ in $\Glexun$ such that $\Inf(\out{\sigma_1, \sigma_2^\ast}_{v}) \lexicoleq_1  (\alpha,\beta) $ for all strategies $\sigma_1$ of player 1.
\end{enumerate}
Moreover, the two partitions of $V$ and the two strategies $\sigma_1^\ast$, $\sigma_2^\ast$ can be computed in polynomial time.
\end{lemma}

\begin{proof}
We only prove the first statement, the second one is proved similarly.
From $\Glexun$, we construct the same arena ${\cal A} = (V', E')$ as in Lemma~\ref{lem:reduction} in a way to have weights depending on vertices instead of edges. Recall that $V' = V \cup E$ such that each edge $e$ of $\Glexun$ is split into two consecutive edges where the new intermediate vertex $s$ is decorated with $(r_1(s), r_2(s))$ and belongs to player 1. The vertices $v$ of $V$ are decorated with $(r_1(v), r_2(v)) = (+\infty, +\infty)$.  

Let $(\alpha, \beta)$ be a pair of rational numbers. We are going to construct a partition of $V'$ into two sets $W_1$ and $W_2$ such that for all $v' \in V'$, we have that $v' \in W_1 \cap V$ iff player 1 has a strategy $\sigma_1^\ast$ (that can be chosen positional) from $v'$ in $\Glexun$ such that $(\alpha,\beta) \lexicoleq_1 \Inf(\out{\sigma_1^\ast, \sigma_2}_{v'})$ for all strategies $\sigma_2$ of player 2. The restriction to $V$ of this partition leads to the required partition of $V$ of Lemma~\ref{lem:partition}. Recall that $(\alpha,\beta) \lexicoleq_1 \Inf(\rho)$, with $\rho = \out{\sigma_1^\ast, \sigma_2}_{v'}$ iff 
\begin{itemize}
\item either $r_1(\rho_n,\rho_{n+1}) > \alpha$  for all $n$,
\item or $r_1(\rho_n,\rho_{n+1}) \geq \alpha$  for all $n$, and $r_2(\rho_n,\rho_{n+1}) \leq \beta$ for one $n$.
\end{itemize}
In the sequel, let us use notation 
$\left\langle\langle i \right\rangle\rangle^{{\cal A}'}\phi$ for the set of vertices $v'$ inside a game restricted to the subarena $\cal A'$ of $\cal A$, from which player $i$ has a strategy that ensures an LTL formula $\phi$, with $\phi$ describing a reachability objective. To get the partition of $V'$ into $W_1$ and $W_2$, we proceed step by step.

\begin{enumerate}
\item We consider $W_2^{< \alpha} = \left\langle\langle 2 \right\rangle\rangle^{\cal A} \Diamond \set{s \in V'  \mid r_1(s)  < \alpha}$, that is, the set of vertices $v'$ in $\cal A$ from which player~2 can ensure to reach a vertex $s$ decorated with $r_1(s) < \alpha$. Clearly, from such vertices $v' \in V \cap W_2^{< \alpha}$ seen as vertices in $\Glexun$, player 1 cannot ensure $(\alpha,\beta) \lexicoleq_1 \Inf(\out{\sigma_1, \sigma_2}_{v'})$ against all strategies $\sigma_2$ of player 2. In other words, $W_2^{< \alpha} \subseteq W_2$. 
\item Let ${\cal A}_1 = {\cal A}|_{V'\setminus W_2^{< \alpha}}$ be the subarena $\cal A$ restricted to vertices in $V^1 = V'\setminus W_2^{< \alpha}$. Notice that in subarena ${\cal A}_1$, every vertex $s \in V^1$ is decorated with $r_1(s) \geq \alpha$ since otherwise $s \in W_2^{< \alpha}$.  We consider the set $W_1^{\leq \beta} = \left\langle\langle 1 \right\rangle\rangle^{{\cal A}_1} \Diamond \set{s \in V^1 \mid r_2(s) \leq \beta}$. We get that $W_1^{\leq \beta} \subseteq W_1$. Moreover, for all $v' \in W_1^{\leq \beta} \cap V$, player 1 has a positional strategy $\sigma_1^\ast$ from $v'$ in $\Glexun$ such that $(\alpha,\beta) \lexicoleq_1 \Inf(\out{\sigma_1^\ast, \sigma_2}_{v'})$. Indeed this strategy is given by the positional strategy that is used in the reachability game played on ${\cal A}_1$ to reach a vertex~$s$ decorated with $r_2(s) \leq \beta$.  
\item We define the subarena ${\cal A}_2 = {\cal A}|_{V^1\setminus W_1^{\leq \beta}}$ of ${\cal A}$ restricted to $V^2 = V^1\setminus W_1^{\leq \beta}$. Each vertex $s$ of ${\cal A}_2$ is decorated with $r_1(s) \geq \alpha$ and $r_2(s) > \beta$. With the set  $W_2^{= \alpha} = \left\langle\langle 2 \right\rangle\rangle^{{\cal A}_2} \Diamond \set{s \in V^2 \mid r_1(s) = \alpha}$, we have that $W_2^{= \alpha} \subseteq W_2$. 
\item Since each vertex $v'$ of $V'\setminus (W_2^{< \alpha} \cup W_1^{\leq \beta} \cup W_2^{= \alpha})$ is decorated with $r_1(v') > \alpha$, it follows that $v'$ belongs to $W_1$. 
Moreover, there exists a positional strategy $\sigma_1^\ast$ of player~1 from $v'$ in $\Glexun$ such that $(\alpha,\beta) \lexicoleq_1 \Inf(\out{\sigma_1^\ast, \sigma_2}_{v'})$. It is given by a positional strategy used in the safety game played on ${\cal A}_2$ to avoid  all vertices $s$ decorated with $r_1(s) = \alpha$. 
\end{enumerate}
Therefore, we get required partition of $V'$ into $W_2 = W_2^{< \alpha} \cup W_2^{= \alpha}$ and $W_1 = V' \setminus W_2$. Moreover, as this partition has been obtained thanks to three reachability games, the sets $W_2^{< \alpha}$, $W_1^{\leq \beta}$ and $W_2^{= \alpha}$ can be computed in $\bigO{|V| + 3 \cdot |E|}$\footnote{The arena $\cal A$ has $|V| + |E|$ vertices and $2 \cdot  |E|$ edges.}, as well as the positional strategy $\sigma_1^\ast$~\cite{LNCS2500}. 
\qed\end{proof}

\begin{proof}[of Theorem~\ref{thm:uniform-det} for $\Inf$ payoff]
Let $\Glexun$ be an $\Inf$ lexicographic payoff game. By Lemma~\ref{lem:valueInf}, we know that $\Glexun$ is determined.
Given an initial vertex $v_0$ with value $(\alpha, \beta)$, let us indicate how to construct positional optimal strategies for both players. We partition $V$ into the sets $W_1$ and $W_2$ (resp. $T_1$ and $T_2$) as indicated in Lemma~\ref{lem:partition}. By Definition~\ref{def:optimal}, we have that $v_0 \in W_1 \cap T_2$, and Lemma~\ref{lem:partition} provides positional strategies $\sigma_1^\ast$ and $\sigma_2^\ast$ that are optimal. Notice that these strategies depend on~$v_0$ since they depend on  $(\alpha,\beta)$. 
\qed\end{proof}

\paragraph{Complexity Results.}

Let us now turn to the proof of Theorem~\ref{thm:compl-lexico}.

\begin{proof}[of Theorem~\ref{thm:compl-lexico} for $\Inf$ payoff]
The second statement is direct consequence of Lemma~\ref{lem:valueInf}. A polynomial time algorithm for the first and third statements follows from Lemma~\ref{lem:partition}.

To establish {\sf P}-completeness for the first statement, we show that for each safety game $\Safety$, we can construct an $\Inf$ lexicographic payoff game $\Glexun_{\Safety}$ and a pair of naturals $(\alpha, \beta)$ such that the problem of deciding if player~$1$ has a winning strategy from a vertex $v_0$ in $\Safety$ is equivalent to the value problem stated in the first statement of Theorem~\ref{thm:compl-lexico} for $\Inf$ payoff. Let $(V,E)$ be the arena of $\Safety$ and $A \subseteq V$ be the set of vertices that player~$1$ wants to avoid. We set $(\alpha, \beta) = (1, 1)$ and we define $\Glexun_{\Safety} = (V \cup \{w\}, V_1 \cup \{w\}, V_2, E' = E \cup \{(w,v_0)\}, r, \Inf)$ with $w \notin V$ and $r$ such that for all $(v, v') \in E'$, $r(v, v') = (0, 1)$ if $v' \in A$, and $r(v, v') = (1, 1)$ otherwise. Clearly, for every play $\rho$ starting in $v_0$ in $\Safety$, we have a corresponding path $w \rho$ in $\Glexun_{\Safety}$ such that $\rho_n \notin A$ for all $n \geq  0$ iff $\Inf(w \rho) = (1, 1)$. 
Therefore, deciding whether player~1 has a winning strategy from $v_0$ in  $\Safety$ is equivalent to decide whether the value of $w$ in $\Glexun_{\Safety}$ is at least equal to $(\alpha, \beta) =(1,1)$.
\qed\end{proof}

\subsection{$\Disc^\lambda$ Lexicographic Payoff Games} {\label{subsec:Disc}

In this section we study the $\Disc^\lambda$ payoff function for which we prove Theorems~\ref{thm:uniform-det} and~\ref{thm:compl-lexico}.

\paragraph{Uniform-Determinacy.}
The proof of Theorem~\ref{thm:uniform-det} for $\Disc^\lambda$ lexicographic payoff games follows from the simple next idea: player 1 first tries to maximize his payoff limited to the first component and then to minimize his payoff limited to the second component. Player 1 can use a uniform strategy to achieve this goal because at each step, the game reduces to a one-dimensional discounted game.

Let us now go into the details. The proof of Theorem~\ref{thm:uniform-det} for $\Disc^\lambda$ payoff will be obtained as a consequence of Lemmas~\ref{lem:disc} and~\ref{lem:optimalDisc} given below.
 
Let $\lambda \in \, ]0,1[$ and $\Glexun = (V, V_1, V_2, E, r, \Disc^\lambda, \lexicoleq_1)$ be a lexicographic payoff game with weight function $\Disc^\lambda$.
We derive from $\Glexun$ the discounted game $\D = (V, V_1, V_2, E, r_1, \Disc^\lambda_1)$ such that the weights are limited to the first component $r_1$ of $r$, and player $1$ wants to maximize the payoff $\Disc^\lambda_1$ while player $2$ wants to minimize it. By \cite{ZP}, this game is uniformly-determined, and the value of 
each vertex $v \in V$ and uniform strategies can be defined by the following system of equations: 
\begin{equation}\label{definitionValue}
\Val(v) = \left\{
  \begin{array}{ll}
    \max_{(v, v') \in E} \set{(1 - \lambda) \cdot r_1(v, v') + \lambda \cdot \Val(v')} & \quad \mbox{if } v \in V_1, \\ 
    \min_{(v, v') \in E} \set{(1 - \lambda) \cdot r_1(v, v') + \lambda \cdot \Val(v')} & \quad\mbox{if } v \in V_2. 
  \end{array}
\right.
\end{equation}

We denote by $E'$ the set of all optimal edges of $E$, i.e. edges realizing the maximum (resp. minimum) for 
player $1$ (resp. player $2$) in system (\ref{definitionValue}). 

The next lemma states that against an optimal strategy of player 1 in $\D$, player 2 should rather choose edges in $E'$. 

\begin{lemma} \label{lem:disc}
\begin{enumerate}
\item Let $\sigma^\star_1$ be a uniform optimal strategy of player 1 in $\D$, and let $\sigma_2$ be a strategy of player 2 in $\D$.  Let $v$ be a vertex in $V$ with value $\Val(v)$ in $\D$. If $\rho = \out{\sigma_1^\star, \sigma_2}_v$ contains two vertices $\rho_l, \rho_{l+1}$ with $(\rho_l, \rho_{l+1}) \in E \setminus E'$, then $\Disc_1^\lambda(\out{\sigma^\star_1, \sigma_2}_v) > \Val(v)$. 
\item Similarly let $\sigma^\star_2$ be a uniform optimal strategy of player 2, and $\sigma_1$ be a strategy of player 1. If  $\rho = \out{\sigma_1, \sigma_2^\star}_v$ contains two vertices $\rho_l, \rho_{l+1}$ with $(\rho_l, \rho_{l+1}) \in E \setminus E'$, then  $\Disc_1^\lambda(\out{\sigma_1, \sigma_2^\star}_v) < \Val(v)$.
\end{enumerate}
\end{lemma}

\begin{proof}
We only give the proof for the first statement, the proof being similar for the second one. Notice that player~1 only uses edges in $E'$.

Suppose that player $2$ follows a strategy $\sigma_2$ such that  $\rho = \out{\sigma_1^\star, \sigma_2}_v$ uses at least one edge in $E \setminus E'$.
We consider the smallest index $l \geq 0$ such that $\rho_l \in V_2$ and $(\rho_l, \rho_{l+1}) \in E \setminus E'$. 

On one hand, we have by (\ref{definitionValue}) 
\begin{eqnarray} \label{eq:val}
\Val(v) &=& (1 - \lambda) \cdot \sum_{k=0}^{l-1} \lambda^k \cdot r_1(\rho_k, \rho_{k+1}) + \lambda^l \cdot \Val(\rho_l). 
\end{eqnarray}

On the other hand, by optimality of $\sigma_1^\star$, we have that 
\begin{eqnarray*}
\Disc^{\lambda}_1\left(\out{\sigma_1^\star, \sigma_2}_v \right) 
                     & = & (1 - \lambda) \cdot \sum_{k=0}^{l} \lambda^k \cdot r_1(\rho_k, \rho_{k+1})  + \lambda^{l+1} \cdot \Disc^{\lambda}_1\left(\out{\sigma_1^\star, {\sigma_2}|_{\rho_{\leq l}}}_{\rho_{l+1}}\right) \\ 
                     &\geq& (1 - \lambda) \cdot \sum_{k=0}^{l} \lambda^k \cdot r_1(\rho_k, \rho_{k+1}) +  \lambda^{l+1} \cdot \Val(\rho_{l+1}).
 \end{eqnarray*}
Then by definition of $l$, we have
\begin{eqnarray}  \label{eq:disc}     
\Disc^{\lambda}_1\left(\out{\sigma_1^\star, \sigma_2}_v \right)                     
& > & (1 - \lambda) \cdot \sum_{k=0}^{l-1} \lambda^k \cdot r_1(\rho_k, \rho_{k+1}) + \lambda^l \cdot 
                                                     \Val (\rho_l).
\end{eqnarray}
By (\ref{eq:val}) and (\ref{eq:disc}), we can conclude that $\Disc^{\lambda}_1\left(\out{\sigma_1^\star, \sigma_2}_v \right) > \Val(v)$. 
\qed\end{proof}

We have shown that if both players play optimally in $\D$, then the produced outcome uses edges of $E'$ only. 
We now consider the discounted game $\D' = (V, V_1, V_2, E', r_2, \Disc^\lambda_2)$ such that the edges are limited to $E'$ and the weight function is now $r_2$. Another difference is that player $1$ wants to \emph{minimize} his payoff (now seen as a cost) while player $2$ wants to \emph{maximize} it. 

The following lemma describes values and optimal strategies of the lexicographic payoff game $\Glexun$ in relation with the ones of discounted games $\D$ and $\D'$. Theorem~\ref{thm:uniform-det} for $\Disc^\lambda$ payoff is a consequence of this lemma.

\begin{lemma}\label{lem:optimalDisc}
\begin{enumerate}
\item Let $v \in V$ be a vertex. If $v$ has value $\Val(v)$ in $\D$ and $\beta(v)$ in $\D'$, 
then $v$ has value $(\Val(v), \beta(v))$ in $\Glexun$. 
\item A uniform optimal strategy in $\D'$ is optimal in $\Glexun$. 
\end{enumerate}
\end{lemma} 

\begin{proof}
Let $\tau_1^\star, \tau_2^\star$ be uniform optimal strategies for player $1$ and player $2$ respectively in $\D'$. Suppose that vertex $v$ has value $\Val(v)$ in $\D$ and value $\beta(v)$ in $\D'$. By Lemma~\ref{lemma:pointDeSelle}, it suffices to show that for all strategies $\sigma_1, \sigma_2$ for player $1$ and player $2$ respectively in $\Glexun$, we have $\Disc^\lambda(\out{\sigma_1, \tau_2^\star}_v) \lexicoleq_1 (\Val(v), \beta(v)) \lexicoleq_1  \Disc^\lambda(\out{\tau_1^\star, \sigma_2}_v)$. 

Let $\rho = \out{\tau_1^\star, \sigma_2}_v$. We are going to show that $(\Val(v), \beta(v)) \lexicoleq_1  \Disc^\lambda(\rho)$. 
As $\tau^\star_1$ is a strategy in $\D'$, it is an optimal strategy in~$\D$. Therefore $\Disc^\lambda_1(\rho) \geq \Val(v)$. If this inequality is strict, then $(\Val(v), \beta(v)) \lexicoleq_1  \Disc^\lambda(\rho)$. Let us thus suppose that $\Disc^\lambda_1(\rho) = \Val(v)$. By Lemma~\ref{lem:disc}, the outcome $\rho$ uses edges of $E'$ only and is then an outcome in $\D'$. By optimality of $\tau_1^\star$ in $\D'$, it follows that $\Disc^\lambda_2(\rho) \leq \beta(v)$, and thus $(\Val(v), \beta(v)) \lexicoleq_1  \Disc^\lambda(\rho)$.

We show similarly that $\Disc^\lambda(\out{\sigma_1, \tau_2^\star}_v) \lexicoleq_1 (\Val(v), \beta(v))$.
\qed\end{proof}

\paragraph{Complexity Results.}
Let us now provide a proof of Theorem~\ref{thm:compl-lexico} for $\Disc^\lambda$ payoff.

\begin{proof} [of Theorem~\ref{thm:compl-lexico} for $\Disc^\lambda$ payoff]
We begin with the second statement. By Lemma~\ref{lem:optimalDisc}, we know that if $v$ has value $\Val(v)$ in $\D$ and $\beta(v)$ in $\D'$, then $v$ has value $(\Val(v), \beta(v))$ in $\Glexun$. The value $\Val(v)$ of each $v$ in the discounted game $\D$ can be computed in pseudo-polynomial time~\cite{ZP}. Once these values are computed in $\D$, the subgame $\D'$ can be computed in polynomial time by substituting each value $\Val(v)$ in system (\ref{definitionValue}) and removing each edge that does not realise the maximum (resp. minimum) for player $1$ (resp. player $2$) in the resulting system. The value $\beta(v)$ of each vertex $v$ in the discounted game $\D'$ can be computed in  pseudo-polynomial time. 

Notice that we also have an algorithm in {\sf NP} $\cap$ {\sf co-NP} to compute the value $(\Val(v), \beta(v))$ of a given vertex $v$ of $\Glexun$. Let us explain the main ideas. Recall that in discounting games, deciding whether the value of a vertex is greater than or equal to a given threshold is in {\sf NP} $\cap$ {\sf co-NP}, and that the solutions to a system like (\ref{definitionValue}) can be written with polynomially many bits~\cite{ZP,LNCS2500}. Hence, there is an exponential number of possibilities for the values of a discounted game. To compute such a value, we can apply a dichotomy that uses the previous decision problem in {\sf NP} $\cap$ {\sf co-NP} a polynomial number of times. We thus have an algorithm that uses a polynomial number of calls to an oracle in  {\sf NP} $\cap$ {\sf co-NP}. Since ${\sf P}^{ \sf NP \cap \sf co\mbox{-}NP} = {\sf NP} \cap$ {\sf co-NP}~\cite{Brassard79}, we get an algorithm in {\sf NP} $\cap$ {\sf co-NP} to compute the value $(\Val(v), \beta(v))$ of vertex $v$ in $\Glexun$.

Let us turn to the first statement of Theorem~\ref{thm:compl-lexico}. Let $(\alpha,\beta)$ be a pair of rational numbers, and let $v$ be a vertex. To check whether $(\alpha,\beta) \lexicoleq_1 (\Val(v), \beta(v))$, we compute $(\Val(v), \beta(v))$ in {\sf NP} $\cap$ {\sf co-NP} as just explained, and then make the comparison with $(\alpha,\beta)$. 

For the third statement, by Lemma~\ref{lem:optimalDisc}, we also know that a uniform optimal strategy in $\D'$ is optimal in $\Glexun$. Uniform optimal strategies in discounted games can be constructed in pseudo-polynomial time~\cite{ZP}, and we have seen before that $\D'$ can also be constructed in pseudo-polynomial time.
\qed\end{proof}

\section{Study of Problems~\ref{prob:ESconstraint} and~\ref{prob:path}} \label{sec:path}

In Section~\ref{sec:lexico}, we have studied the determinacy of lexicographic payoff games and the related complexities. 
From these results we were able to prove Theorems~\ref{thm:existenceES} and~\ref{thm:complES} about the existence and the construction of a secure equilibrium in weighted games with the payoff functions of Definition~\ref{def:payoff}.

The aim of this section is to provide a proof of Theorem~\ref{thm:constraintES} about the constrained existence of a secure equilibrium in such games (see Problem~\ref{prob:ESconstraint}). Thanks to Proposition~\ref{constSE} that provides a general framework for solving Problem~\ref{prob:ESconstraint}, it remains to study Problem~\ref{prob:path} about the constraint existence of a path in a graph with weights and values.
Hence we first study the latter problem. We then derive a proof of Theorem~\ref{thm:constraintES} for $\MPInf, \MPSup, \LimInf,\LimSup,\Inf$ and $\Sup$ payoffs. We leave the discounted case open. However we show that a solution in this case would provide
a solution to an open problem mentioned in \cite{ChatterjeeFW13} itself related to other difficult open problems in mathematics \cite{Boker}.

\subsection{Solution for Problem~\ref{prob:path}}

Let us first show that Problem~\ref{prob:path} is decidable for $\MPInf$, $\MPSup$, $\LimInf$ and $\LimSup$ payoffs. The approach that we develop is inspired by proof techniques proposed in~\cite{UmmelsW11}.

\begin{theorem} \label{thm:path}
Let $G = (V,E,v_0,r,\Val)$ be a finite directed graph with an initial vertex $v_0$, a weight function $r$, and a value function $\Val$. Let $\mu, \nu \in (\IQ \cup \{\pm \infty\})^2$ be two thresholds. Then for $\g = \MPInf$, $\MPSup$, $\LimInf$ and $\LimSup$, one can decide in polynomial time whether there exists an infinite path $\rho$ in $G$ such that 
$\forall k \geq 0, \forall i \in \{1,2\}, \Val^i(\rho_k) \lexicoleq_i \g(\rho_{\geq k})$, and
$\mu \leq \g(\rho) \leq \nu$.
\end{theorem}

The proof of this theorem is based on the next lemma.
\begin{lemma} \label{lem:ummels}
One can decide in polynomial time whether there exists an infinite path $\rho$ in $G$ such that for each $i \in \{1,2\}$, $\mu_i \sim_i \g_i(\rho) \sim'_i \nu_i$ with $\sim_i, \sim'_i \, \in \{<, \leq\}$.
\end{lemma}

\begin{proof}
We begin with $\MPInf$ payoff (the proof can be easily adapted to $\MPSup$ payoff).
Lemma~\ref{lem:ummels} is stated in~\cite{UmmelsW11} and proved in~\cite{UmmelsArXiv}, but for $\sim_i, \sim'_i$ equal to $\leq$ only. The proposed proof reduces the existence of the required path $\rho$ to the existence of a solution of a linear program 
(there is no linear objective function to optimize, just a finite set of linear constraints to solve). One can check in polynomial time whether there exists a solution to a linear program~\cite{Schrijver}.
In our case, the proof of~\cite{UmmelsArXiv} leads to a set of linear constraints with both strict and non strict inequalities. One can also check in polynomial time the existence of a solution of such a set of constraints (see Lemma~\ref{lem:strict} in the appendix).

In the case of $\LimInf$ payoff, let us show a reduction to the emptiness problem for Rabin automata (the proof is similar for $\LimSup$ payoff). Let $\rho$ be a path in $G$ starting in $v_0$ such that for each $i \in \{1,2\}$, $\mu_i \sim_i \LimInf_i(\rho) \sim'_i \nu_i$. Equivalently, for each $i \in \{1,2\}$, there exists $n_i \geq 0$ such that $\mu_i \sim_i r_i(\rho_k,\rho_{k+1})$ for all $k \geq n_i$, and $r_i(\rho_l,\rho_{l+1}) \sim'_i \nu_i$ for infinitely many $l$. This is the conjonction of two Rabin conditions, one for each component $i$. 
From $G$, we construct a Rabin automaton ${\cal R}_i$, $i \in \{1,2\}$, with initial vertex $v_0$, and sets of vertices and edges $V'_i,E'_i$ defined as follows. Each edge $e = (v,v') \in E$ is split into two consecutive edges, and the new intermediate vertex is decorated with $r_i(e)$. The set $V'_i$ has thus $|V| + |E|$ vertices, such that vertices of $V$ are decorated with $+\infty$. We then define one Rabin pair $(A_i,B_i)$ such that $A_i \subseteq V'_i$ (resp. $B_i \subseteq V'_i$) is composed of all vertices decorated by $a$ such that $\neg(\mu_i \sim_i a)$  (resp. by $b$ such that $b \sim'_i \nu_i$).
Consider the Rabin automaton $\cal R$ being the intersection of the automata ${\cal R}_1$ and ${\cal R}_2$. We have that there exists an infinite path $\rho$ in $G$ such that $\forall i\in \{1,2\}, \mu_i \sim_i \g_i(\rho) \sim'_i \nu_i$ if and only if there exists an accepting path in $\cal R$. The latter property can be checked in polynomial time~\cite{KingKV01}.
\qed\end{proof}

\begin{proof}[of Theorem~\ref{thm:path}]
Notice that each $\g_i$ is prefix-independent in Theorem~\ref{thm:path} (see Remark~\ref{rem:prefix-linear}). Hence condition 
$\forall k \geq 0, \forall i \in \{1,2\}, \Val^i(\rho_k) \lexicoleq_i \g(\rho_{\geq k})$ can be replaced by
\begin{eqnarray} \label{eq:caractESbis}
\forall k \geq 0, \forall i \in \{1,2\}, \quad \Val^i(\rho_k) \lexicoleq_i \g(\rho).
\end{eqnarray}
We thus propose Algorithm~\ref{algo} to solve Problem~\ref{prob:path}.
\begin{algorithm}[t]
\caption{$Constrained Existence(G,\mu,\nu)$}
\label{algo}
\begin{algorithmic}[1]
\FOR {each $s_1, s_2 \in V$}
\STATE Compute the subgraph $G'$ of $G$ whose vertices $v$ are such that $\Val^i(v) \lexicoleq_i \Val^i(s_i)$, for $i \in \{1,2\}$ 
\IF {there exists a play $\rho$ in $G'$ such that $\Val^i(s_i) \lexicoleq_i \g(\rho)$, for $i \in \{1,2\}$, and $\mu \leq \g(\rho) \leq \nu$}
\RETURN True
\ENDIF
\ENDFOR
\RETURN False
\end{algorithmic}
\end{algorithm}

To prove its soundness, assume that the algorithm returns True. Hence, there exist a subgraph $G'$ of $G$ depending on two vertices $s_1, s_2$, and a play $\rho$ in $G'$ such that $\Val^i(s_i) \lexicoleq_i \g(\rho)$, for $i \in \{1,2\}$, and $\mu \leq \g(\rho) \leq \nu$. It follows that $\rho$ is a play in $G$ that satisfies ($\ref{eq:caractESbis}$) and $\mu \leq \g(\rho) \leq \nu$.

To prove that the algorithm is complete, let $\rho$ be a play that satisfies ($\ref{eq:caractESbis}$) and $\mu \leq \g(\rho) \leq \nu$. For each $i$, let $s_i \in V$ be such that $\Val^i(s_i)$ is the maximum value among $\Val^i(\rho_k)$, $k \geq 0$ (for the order $\lexicoleq_i$). Therefore $\Val^i(s_i) \lexicoleq_i \g(\rho)$, for $i \in \{1,2\}$. It follows that $\rho$ is a play in the subgraph $G'$ of $G$ whose vertices $v$ are such that $\Val^i(v) \lexicoleq_i \Val^i(s_i)$, for $i \in \{1,2\}$. As $\mu \leq \g(\rho) \leq \nu$, the algorithm will return True.

Let us now explain how to check the existence of a path as indicated in line~3 of Algorithm~\ref{algo}. If one recalls the definition of the orders $\lexicoleq_1$ and $\lexicoleq_2$, one notices that the condition that a path $\rho$ has to satisfy in line~3 is equivalent to the disjunction of four conditions of the form $\forall i \in \{1,2\}, x_i \sim_i \g_i(\rho) \sim'_i y_i$ with $x, y \in (\IQ \cup \{\pm \infty\})^2$ and $\sim_i, \sim'_i \, \in \{<, \leq\}$. The latter conditions can be checked in polynomial time by Lemma~\ref{lem:ummels}.

Hence, Algorithm~\ref{algo} is correct and works in polynomial time.
\qed\end{proof}

\subsection{Solution to Problem~\ref{prob:ESconstraint}}

We now have all the required material to prove Theorem~\ref{thm:constraintES} for all the payoffs except $\Disc^\lambda$ payoff. We begin with $\MPInf$, $\MPSup$, $\LimInf$ and $\LimSup$ payoffs, since $\Inf$ and $\Sup$ payoffs require a separate proof.

\begin{proof} [of Theorem~\ref{thm:constraintES} for $\MPInf$, $\MPSup$, $\LimInf$ and $\LimSup$ payoffs]
By Remark~\ref{rem:prefix-linear}, and Theorems~\ref{thm:uniform-det},~\ref{thm:compl-lexico} and~\ref{thm:path}, we see that the hypotheses of Proposition~\ref{constSE} are satisfied. Therefore, given an initialized weighted game $(\G, v_0)$ and two thresholds $\mu, \nu \in (\mathbb{Q} \cup \{\pm \infty\})^2$, one can decide whether there exists a secure equilibrium $(\sigma_1, \sigma_2)$ in $(\G, v_0)$ such that $\mu \leq \g(\out{\sigma_1, \sigma_2}_{v_0}) \leq \nu$. Let us come back to the algorithm that was proposed to solve this problem and let us study its complexity (see the proof of Proposition~\ref{constSE}): (1) we compute $\Val^i(v)$ for each vertex $v$ of $\Glexi$, $i \in \{1,2\}$; (2) we construct from $(\G,v_0)$ the graph $G = (V,E,v_0,r,\Val)$ on which Algorithm~\ref{algo} is applied. For each payoff $\MPInf, \MPSup, \LimInf$ and $\LimSup$, step (2) can be done in polynomial time by Theorem~\ref{thm:path}. For $\LimInf$ and $\LimSup$ payoffs, step (1) can also be done in polynomial time by Theorem~\ref{thm:compl-lexico}. 

To complete the proof, it remains to show that step (1) is in {\sf NP} $\cap$ {\sf co-NP} for $\MPInf$ and $\MPSup$ payoffs. By Theorem~\ref{thm:compl-lexico}, we know that (a) optimal strategies are uniform for $\MPInf$ and $\MPSup$ lexicographic payoff games, and (b) deciding whether the value of a vertex is greater than or equal to a given threshold is in {\sf NP} $\cap$ {\sf co-NP}.
Therefore, by (a) the value of a vertex is a bounded rational (between $0$ and $|R|$\footnote{Under Remark~\ref{rem:positive}, recall that $|R|$ is the maximal weight of the game.}) whose denominator is at most equal to $|V|$ (the maximal length of a simple cycle), thus leading to an exponential number of possibilities. To compute this value, we can apply a simple dichotomy that uses (b). We thus have an algorithm using a polynomial number of calls to an oracle in  {\sf NP} $\cap$ {\sf co-NP}. Since ${\sf P}^{ \sf NP \cap \sf co\mbox{-}NP} = {\sf NP} \cap$ {\sf co-NP}~\cite{Brassard79}, step (1) is in {\sf NP} $\cap$ {\sf co-NP} as announced.
\qed\end{proof}

\begin{proof} [of Theorem~\ref{thm:constraintES} for $\Inf$ and $\Sup$ payoffs]
We use the same reduction to $\LimInf$ and $\LimSup$ weighted games as done in the proof of Theorems~\ref{thm:existenceES} and~\ref{thm:complES} for $\Inf$ and $\Sup$ payoffs. As Theorem~\ref{thm:constraintES} holds for $\LimInf$ and $\LimSup$ payoffs, it also holds for $\Inf$ and $\Sup$ payoffs.
\qed\end{proof}

\subsection{Particular Case of $\Disc^\lambda$ Payoff}

In this section, we consider $\Disc^\lambda$ weighted games. We are going to show that Problem~\ref{prob:ESconstraint} for these games is related to the next problem whose decidability is unknown~\cite{ChatterjeeFW13}. Moreover the latter problem is related to other hard open problems in diverse mathematical fields according to~\cite{Boker}. 

\begin{problem} \label{prob:Boker}
Given three rational numbers $a,b$ and $t$, and a rational discount factor $\lambda \in \, ]0,1[$, does there exist an infinite sequence $w = w_0w_1 \ldots \in \{a, b\}^\omega$ such that $\sum_{k=0}^{\infty} w_k \lambda^k$ is equal to $t$?
\end{problem}

Let us provide the next reduction of Problem~\ref{prob:Boker} to Problem~\ref{prob:ESconstraint} for $\Disc^{\lambda}$ payoff. Let $a,b, t \in \IQ$ and a rational factor $\lambda \in ]0,1[$. We consider the $\Disc^\lambda$ weighted game $(\G^{a,b},v_0)$ played on the arena depicted in Figure~\ref{Figure:reduction}, where $c$ is a rational number such that 
$c > \max \set{|a|, |b|}$.\footnote{Notice that $(\G^{a,b},v_0)$ have mutiple edges between two vertices. Adequate intermediate vertices should be added to respect Definition~\ref{RG}.}  We also consider the thresholds $\mu = ((1-\lambda)\cdot t,- (1-\lambda)\cdot t)$ and $\nu = (+\infty,+\infty)$. Notice that threshold $\nu$ imposes no constraint. We want to show the next proposition:

\begin{proposition} \label{prop}
There exists a sequence $w \in \{a, b\}^\omega$  such that $\sum_{k=0}^{\infty} w_k \lambda^k = t$ if and only if there exists a secure equilibrium in $(\G^{a,b},v_0)$ with outcome~$\rho$ such that $\mu \leq \Disc^{\lambda}(\rho) \leq \nu$.
\end{proposition}

\begin{figure}[ht!]
\begin{center}
\begin{tikzpicture}[initial text=,auto, node distance=3cm, shorten >=1pt] 

\node[state, initial, scale=0.7]    (v0)                           {$v_0$};
\node[state, rectangle, scale=0.7]    (v1)         [right=of v0]     {$v_1$};
\node (fictif) [right=1.5cm of v0] {};
\node[state, scale=0.7]               (v2)    [below=1.5 of fictif]    {$v_2$};

\path[->] (v0) edge [bend left=45]              node[midway, scale=0.8] {$(b, -b)$} (v1)
               edge [bend left=15]              node[midway, scale=0.8] {$(a, -a)$} (v1)
               edge [bend left=-25]              node[left, scale=0.8] {$(-c, -c)$} (v2)

          (v1) edge [bend right=-45]  node[midway, scale=0.8] {$(a, -a)$}  (v0)
               edge [bend right=-15]              node[midway, scale=0.8] {$(b, -b)$} (v0)
               edge [bend right=-25]              node[right, scale=0.8] {$(-c, -c)$} (v2)

          (v2) edge [loop below]  node[midway, scale=0.8] {$(-c, -c)$}  ();

\end{tikzpicture}
\end{center}
\caption{The $\Disc^\lambda$ weighted game $\G^{a,b}$ \label{Figure:reduction}}
\end{figure}
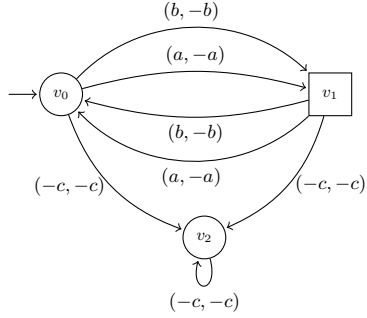

In the next lemma, we denote by $\A$ the subarena of $\G^{a,b}$ restricted to 
the set of vertices $\set{v_0, v_1}$.

\begin{lemma}\label{lem:path<->SE}
For all infinite paths $\rho$ of $\G^{a,b}$ starting in $v_0$,  $\rho$ is a path in $\A$ if and only if 
there exists a secure equilibrium $(\sigma_1, \sigma_2)$ in $(\G^{a,b}, v_0)$ 
with outcome~$\rho$.
\end{lemma}

\begin{proof}
Suppose that $\rho$ is a path in  $\A$. We constuct the strategy profile $(\sigma_1, \sigma_2)$ 
as follows. Both strategies follow $\rho$, and as soon as a player deviates from this path, the other player 
immediately deviates to vertex $v_2$. It is clear that no player has an incentive to deviate. Indeed both players 
want to maximize their discounted sum, and the self loop on vertex $v_2$ is labelled with $-c$ which is strictly smaller 
than all the weights in the subarena $\A$. So no deviation from $\rho$ is profitable, 
and $(\sigma_1, \sigma_2)$ is a secure equilibrium. 

Suppose now that $\rho$ is the outcome of a secure equilibrium 
$(\sigma_1, \sigma_2)$ in $(\G^{a,b}, v_0)$. By contradiction, assume 
that $\rho$ visits vertex $v_2$. Consider the smallest $k$ 
such that $\rho_{k+1} = v_2$. Without lost of generality suppose that $\rho_k = v_0$.
Then player~1 has a strategy $\sigma'_1$ that goes to $v_1$ and thus
avoids going to  $v_2$ (for at least one more round). This is a profitable deviation for him, because all the edges that go to $v_2$ have reward $(-c, -c)$, which is 
strictly smaller than rewards in the subarena $\A$. 
This is in contradiction with $(\sigma_1, \sigma_2)$ being a secure equilibrium. 
\qed\end{proof}

\begin{proof}[of Proposition~\ref{prop}]
Lemma~\ref{lem:path<->SE} states that paths $\rho$ in $\A$ that start in $v_0$ are exactly outcomes of secure equilibria in $(\G^{a,b},v_0)$. Moreover, by definition of the weight function of $\G^{a,b}$, we have $\Disc^{\lambda}_1(\rho) = - \Disc^{\lambda}_2(\rho)$. Therefore there exists a sequence $w \in \{a, b\}^\omega$  such that $\sum_{k=0}^{\infty} w_k \lambda^k = t$ if and only if there exists a secure equilibrium in $(\G^{a,b},v_0)$ with outcome~$\rho$ such that $\mu \leq \Disc^{\lambda}(\rho) ~(\leq \nu)$.
\qed\end{proof}

We have thus shown that Problem~\ref{prob:Boker} reduces to Problem~\ref{prob:ESconstraint} for $\Disc^{\lambda}$ payoff. Therefore Problem~\ref{prob:Boker} is decidable if Problem~\ref{prob:ESconstraint} is itself decidable for this payoff.

\section{Conclusion} \label{sec:conclusion}

In this paper, we have studied Problems~\ref{prob:ES}-\ref{prob:ESconstraint} about secure equilibria in two-player weighted games, and proposed three general frameworks in which we can solve them. We have proved that weighted games with a classical payoff function like $\MPInf$, $\MPSup$, $\LimInf$, $\LimSup$, $\Inf$, $\Sup$ and $\Disc^{\lambda}$, all fall in these frameworks, except the $\Disc^{\lambda}$ payoff for Problem~\ref{prob:ESconstraint} (see Tables~\ref{table:existence}-~\ref{table:constrained}). We have shown that this particular problem is linked to another challenging open problem~\cite{Boker,ChatterjeeFW13}.

Our approach was inspired by the recent work~\cite{TJS} that states the existence of Nash equilibria in a large class of multi-player weighted games. As in~\cite{TJS}, we have considered two particular zero-sum games $\Glexun$, $\Glexdeux$ associated with the initial game $\G$; we have proved that they are uniformly-determined and studied their complexity for all the classical payoffs (see Table~\ref{table:determined}). These results are very interesting on their own right, and to the best of our knowledge, were not studied in the literature, except for a variant of the $\MPInf$ payoff studied in~\cite{KTB}. For $\MPInf$ payoff, our proofs were inspired by techniques developed in~\cite{KTB} and~\cite{UmmelsW11}, however with far from trivial adaptations.

Independently of our paper, in~\cite{Julie14}, the authors give general hypotheses on multi-player weighted games that guarantee the existence of a secure equilibrium. These hypotheses are satisfied by the weighted games studied in this paper with $\LimInf$, $\LimSup$, $\Inf$, $\Sup$ and $\Disc^{\lambda}$ payoffs (but not with $\MPInf$ and $\MPSup$ payoffs). Thus for these payoffs, Problem~\ref{prob:ES} is also positively solved thanks to the results proved in~\cite{Julie14} (however with no indication about the existence of a secure equilibrium that is finite-memory). 

\section*{Acknowledgments}
We thank Jean Cardinal, Martine Labb\'e, Quentin Menet,  and   David Sbabo for fruitful discussions.

\bibliographystyle{abbrv}
\bibliography{se_main.bib}

\section*{Appendix}

\begin{lemma}\label{lem:strict}
Let 
$$\left\{
    \begin{array}{llll}
        \bar{A}_i \cdot \bar{x} & >& b_i & \quad 1 \leq i \leq n \\ 
        \bar{A}_{n + j} \cdot \bar{x} & \geq& b_{n+j} &\quad 1 \leq j \leq m
            \end{array}
            \right.
$$
be a system of $n$ strict linear constraints and $m$ non strict linear constraints.
Then one can decide in polynomial time whether this system has a solution.
\end{lemma}

\begin{proof}
The problem of deciding if the previous system $S$ has a solution can be reduced to the following linear program $P$  where $y$ is a new variable: 
$$\max y $$
$$\left\{
  \begin{array}{llll}
    \bar{A}_i \cdot \bar{x} - y& \geq b_i & \quad 1 \leq i \leq n \\ 
    \bar{A}_{n + j} \cdot \bar {x}  & \geq b_{n+j} &\quad 1 \leq j \leq m
  \end{array}
\right.
$$ 
If $P$ has no solution, then $S$ has no solution neither. If $P$ has a solution, depending on whether the optimal value $y$ satisfies $y > 0$ or $y \leq 0$, then $S$ has a solution or not. 
\qed\end{proof}

\end{document}